\documentclass[11pt]{article}
\usepackage{fullpage}
\usepackage{amssymb}
\usepackage{amsmath,amsfonts,amsthm,amssymb}
\usepackage{enumerate}
\usepackage{setspace}
\usepackage{fancyhdr}
\usepackage{lastpage}
\usepackage{extramarks}
\usepackage{algorithm}
\usepackage{algorithmic}
\usepackage{chngpage}
\usepackage{soul,color}
\usepackage{graphicx,float,wrapfig}
\usepackage[margin=1in]{geometry}

\usepackage[colorlinks,citecolor={black},urlcolor={black},linkcolor={black}]{hyperref}

\usepackage[english]{babel}
\usepackage[colorinlistoftodos]{todonotes}
\usepackage{bbm}

\newtheorem{theorem}{Theorem}[section]
\newtheorem{corollary}[theorem]{Corollary}

\newtheorem{definition}[theorem]{Definition}
\newtheorem{lemma}[theorem]{Lemma}

\newcommand{\D}{\mathcal{D}}
\newcommand{\eps}{\varepsilon}

\DeclareMathOperator{\E}{\mathbb{E}}
\def \Rev  {{\sf Rev}}
\def \Reg  {{\sf Reg}}
\def \cont  {{\sf cont}}
\def \supp {{\sf supp}}
\def \Val {{\sf Val}}
\def \Mye {{\sf Mye}}
\def \Crit {{\sf MBRev}}
\def \lowregret {{no-regret }}

\title{Selling to a No-Regret Buyer}
\author{
Mark Braverman \thanks{Department of Computer Science, Princeton University, email: mbraverm@cs.princeton.edu. Research supported in part by an NSF CAREER award (CCF-1149888), NSF CCF-1215990, NSF CCF-1525342, NSF CCF-1412958, a Packard Fellowship in Science and Engineering, and the Simons Collaboration on Algorithms and Geometry.}
\and
Jieming Mao  \thanks{Department of Computer Science, Princeton University, email: jiemingm@cs.princeton.edu.}
\and
Jon Schneider \thanks{Department of Computer Science, Princeton University, email: js44@cs.princeton.edu}
\and 
S. Matthew Weinberg \thanks{Department of Computer Science, Princeton University, email: smweinberg@princeton.edu. Supported by NSF CCF-1717899. }
}

\begin{document}
\begin{titlepage}
\maketitle \thispagestyle{empty}
\begin{abstract}

We consider the problem of a single seller repeatedly selling a single item to a single buyer (specifically, the buyer has a value drawn fresh from known distribution $\D$ in every round). Prior work assumes that the buyer is fully rational and will perfectly reason about how their bids today affect the seller's decisions tomorrow. In this work we initiate a different direction: the buyer simply runs a no-regret learning algorithm over possible bids. We provide a fairly complete characterization of optimal auctions for the seller in this domain. Specifically:
\begin{itemize}
\item If the buyer bids according to EXP3 (or any ``mean-based'' learning algorithm), then the seller can extract expected revenue arbitrarily close to the expected welfare. This auction is independent of the buyer's valuation $\D$, but somewhat unnatural as it is sometimes in the buyer's interest to overbid. 
\item There exists a learning algorithm $\mathcal{A}$ such that if the buyer bids according to $\mathcal{A}$ then the optimal strategy for the seller is simply to post the Myerson reserve for $\D$ every round. 
\item If the buyer bids according to EXP3 (or any ``mean-based'' learning algorithm), but the seller is restricted to ``natural'' auction formats where overbidding is dominated (e.g. Generalized First-Price or Generalized Second-Price), then the optimal strategy for the seller is a pay-your-bid format with decreasing reserves over time. Moreover, the seller's optimal achievable revenue is characterized by a linear program, and can be unboundedly better than the best truthful auction yet simultaneously unboundedly worse than the expected welfare. 
\end{itemize} \end{abstract}
\end{titlepage}

\section{Introduction}

Consider a bidder trying to decide how much to bid in an auction (for example, a sponsored search auction). If the auction happens to be the truthful Vickrey-Clarke-Groves auction~\cite{Vickrey61,Clarke71,Groves73}, then the bidder's decision is easy: simply bid your value. If instead, the bidder is participating in a Generalized First-Price (GFP) or Generalized Second-Price (GSP) auction, the optimal strategy is less clear. Bidders can certainly attempt to compute a Bayes-Nash equilibrium of the associated game and play accordingly, but this is unrealistic due to the need for accurate priors and extensive computation.

Alternatively, the bidders may try to learn a best-response over time (possibly offloading the learning to commercial bid optimizers). We specifically consider bidders who \emph{\lowregret learn}, as empirical work of Nekipelov et al.~\cite{NekipelovST15} shows that bidder behavior on Bing is largely consistent with \lowregret learning (i.e. for most bidders, there exists a per-click value such that their behavior guarantees \lowregret for this value). From the perspective of a revenue-maximizing auction designer, this motivates the following question: \textbf{If a seller knows that buyers are no-regret learning over time, how should they maximize revenue?}  

This question is already quite interesting even when there is just a single item for sale to a single buyer.\footnote{And also surprisingly relevant: search engines don't generally publish their formulas for setting reserves. So even if you are the only bidder for a certain keyword (e.g. the name of your new startup), you're likely participating in a GSP/GFP auction with no additional bidders, but against a seller who adaptively sets the reserve price based on past bids. Anecdotal evidence indeed suggests that the reserve prices in such single-bidder auctions will change over time.} We consider a model where in every round $t$, the seller solicits a bid $b_t \in [0,1]$ from the buyer, then allocates the item according to some allocation rule $x_t(\cdot)$ and charges the bidder according to some pricing rule $p_t(\cdot)$ (satisfying $p_t(b) \leq b\cdot x_t(b)$ for all $t, b$). Note that the allocation and pricing rules (henceforth, auction) can differ from round to round, and that the auction need not be truthful. Each round, the bidder has a value $v_t$ drawn independently from $\D$, and uses some no-regret learning algorithm to decide which bid to place in round $t$, based on the outcomes in rounds $1,\ldots, t-1$ (we will make clear exactly what it means for a buyer with changing valuation to play no-regret in Section~\ref{sec:prelim}, but one can think of $v_t$ as providing a ``context'' for the bidder during round $t$). 

One default strategy for the seller is to simply to set Myerson's revenue-optimal reserve price for $\D$, $r(\D)$, in every round (that is, $x_t(b_t) = I(b_t \geq r(\D))$, $p_t(b_t) = r(\D) \cdot I(b_t \geq r(\D))$ for all $t$, where $I(\cdot)$ is the indicator function). It's not hard to see that any no-regret learning algorithm will eventually learn to submit a winning bid during all rounds where $v_t > r(\D)$, and a losing bid whenever $v_t < r(\D)$. So if $\Rev(\D)$ denotes the expected revenue of the optimal reserve price when a single buyer is drawn from $\D$, the default strategy guarantees the seller revenue $T\cdot \Rev(\D) - o(T)$ over $T$ rounds. The question then becomes whether or not the seller can beat this benchmark, and if so by how much.

The answer to this question isn't a clear-cut yes or no, so let's start with the following instantiation: how much revenue can the seller extract if the buyer runs EXP3~\cite{AuerCNS03}? In Theorem~\ref{thm:nc_seller}, we show that the seller can actually do \emph{much} better than the default strategy: it's possible to extract revenue per round equal to (almost) the full expected welfare! That is, if $\Val(\D) = \mathbb{E}_{v \leftarrow \D}[v]$, there exists an auction that extracts revenue $T \cdot \Val(\D)-o(T)$ for all $\D$.\footnote{The order of quantifiers in this sentence is correct: it is actually the same auction format that works for all $\D$.} It turns out this result holds not only for EXP3, but for any learning algorithm with the following (roughly stated) property: if at time $t$, the mean reward of action $a$ is significantly larger than the mean reward of action $b$, the learning algorithm will choose action $b$ with negligible probability. We call a learning algorithm with this property a ``mean-based'' learning algorithm and note that many commonly used learning algorithms - EXP3, Multiplicative Weights Update~\cite{AroraHK12}, and Follow-the-Perturbed-Leader~\cite{Hannan57,KalaiV02,KalaiV05} - are `mean-based' (see Section~\ref{sec:prelim} for a formal definition).

We postpone all intuition until Section~\ref{sec:mbaa} with a worked-through example, but just note here that the auction format is quite unnatural: it ``lures'' the bidder into submitting high bids early on by giving away the item for free, and then charging very high prices (but still bounded in $[0,1]$) near the end. The transition from ``free'' to ``high-price'' is carefully coordinated across different bids to achieve the revenue guarantee.

This result motivates two further directions. First, do there exist other no-regret algorithms for which full surplus extraction is impossible for the seller? In Theorem~\ref{thm:lowregret}, we show that the answer is yes. In fact, there is a simple no-regret algorithm $\mathcal{A}$, such that when the bidder uses algorithm $\mathcal{A}$ to bid, the default strategy (set the Myerson reserve every round) is optimal for the seller. We again postpone a formal statement and intuition to Section~\ref{sec:better}, but just note here that the algorithm is a natural adaptation of EXP3 (or in fact, any existing no-regret algorithm) to our setting.

Finally, it is reasonable to expect that bidders might use off-the-shelf no-regret learning algorithms like EXP3, so it is still important to understand what the seller can hope to achieve if the buyer is specifically using such a ``mean-based'' algorithm (formal definition in Section~\ref{sec:prelim}). Theorem~\ref{thm:nc_seller} is perhaps unsatisfying in this regard because the proposed auction is so unnatural, and looks nothing like the GSP or GFP auctions that initially motivated this study. It turns out that the key property separating GFP/GSP from the unnatural auction above is whether overbidding is a dominated strategy. That is, in our unnatural auction, if the bidder truly hopes to guarantee low regret they must seriously consider overbidding (and this is how the auction lures them into bidding way above their value). In both GSP and GFP, overbidding is dominated, so the bidder can guarantee no regret while overbidding with probability $0$ in every round. 

The final question we ask is the following: if the buyer is using EXP3 (or any ``mean-based'' algorithm), but only considering undominated strategies, how much revenue can the seller extract using an auction where overbidding is dominated in every round? It turns out that the auctioneer can still outperform the default strategy, but not extract full welfare. Instead, we identify a linear program (as a function of $\D$) that tightly characterizes the optimal revenue the seller can achieve in this setting when the buyer's values are drawn from $\D$. Moreover, we show that the auction that achieves this guarantee is natural, and can be thought of as a first-price auction with decreasing reserves over time. Finally, we show that this ``mean-based revenue'' benchmark, $\Crit(\D)$ lies truly in between the Myerson revenue and the expected welfare: for all $c$, there exists a distribution $\D$ over values such that $c \cdot T \cdot \Rev(\D) < \Crit(\D) < \frac{1}{c}\cdot T \cdot \Val(\D)$. In other words, the seller's mean-based revenue may be unboundedly better than the default strategy, yet simultaneously unboundedly far from the expected welfare. We provide formal statements and a detailed proof overview of these results in Section~\ref{sec:minicritical}. To briefly recap, our main results are the following:
\begin{enumerate}
\item If the buyer uses a ``mean-based'' learning algorithm like EXP3, the seller can extract revenue $(1-\varepsilon)T\cdot \Val(\D)-o(T)$ for any constant $\varepsilon > 0$ (Theorem~\ref{thm:nc_seller}).
\item There exists a natural no-regret algorithm $\mathcal{A}$ such that when the buyer bids according to $\mathcal{A}$, the seller's default strategy (charging the Myerson reserve every round) is optimal (Theorem~\ref{thm:lowregret}).
\item If the buyer uses a ``mean-based'' algorithm only over undominated strategies, the seller can extract revenue $\Crit(\D)$ using an auction where overbidding is dominated in every round. Moreover, we characterize $\Crit(\D)$ as the value of a linear program, and show it can be simultaneously unboundedly better than $T\cdot \Rev(\D)$ and unboundedly worse than $T\cdot \Val(\D)$ (Theorems~\ref{thm:mbrevlbmain}, \ref{thm:mbrevubmain} and \ref{thm:mbrev}). 
\end{enumerate}

Our plan for the remaining sections is as follows. Below, we overview our connection to related work. Section~\ref{sec:prelim} formally defines our model. Section~\ref{sec:example} works through a concrete example, providing intuition for all three results. 
Section~\ref{sec:conclusions} discusses conclusions and open problems. 

\subsection{Related Work}

There are two lines of work that are most related to ours. The first is that of \emph{dynamic auctions}, such as~\cite{PapadimitriouPPR16,AshlagiDH16,MirrokniLTZ16,MirrokniLTZ16b,LiuP17}. Like our model, there are $T$ rounds where the seller has a single item for sale to a single buyer, whose value is drawn from some distribution every round. However, the buyer is fully strategic and processes fully how their choices today affect the seller's decisions tomorrow (e.g. they engage with deals of the form ``pay today to get the item tomorrow''). Additional closely related work is that of Devanur et al. studying the Fishmonger problem~\cite{DevanurPS15, ImmorlicaLPT17}. Here, there is again a single buyer and seller, and $T$ rounds of sale. Unlike our model, the buyer draws a value from $\D$ once during round $0$ and that value is fixed through all $T$ rounds (so the seller could try to learn the buyer's value over time). Also unlike our model, they study perfect Bayesian equilibria (where again the buyer is fully strategic, and reasons about how their actions today affect the seller's behavior tomorrow). 

In contrast to these works, while buyers in our model do care about the future (e.g. they value learning), they don't reason about how their actions today might affect the seller's decisions tomorrow. Our model is more realistic for sponsored search auctions, where search engines rarely release proprietary algorithms for setting reserves based on past data (and fully strategic reasoning is simply impossible without the necessary information). 

Other related work considers the \emph{Price of Anarchy} of simple combinatorial auctions when bidders no-regret learn~\cite{Roughgarden12, SyrgkanisT13, NekipelovST15, DaskalakisS16}. One key difference between this line of work and ours is that these all study welfare maximization for combinatorial auctions with rich valuation functions. In contrast, our work studies revenue maximization while selling a single item. Additionally, in these works the seller commits to a publicly known auction format, and the only reason for learning is due to the strategic behavior of other buyers. In contrast, buyers in our model have to learn \emph{even when they are the only buyer}, due to the strategic nature of the seller.

Recent work has also considered learning from the perspective of the seller. In these works, the buyer's (or buyers') valuations are drawn from an unknown distribution, and the seller's goal is to learn an approximately optimal auction with as few samples as possible~\cite{ColeR14,DevanurHP16,MorgensternR15,MorgensternR16,GonczarowskiN17,CaiD17,DudikHLSSV17}. These works consider numerous different models and achieve a wide range of guarantees, but all study the learning problem from the perspective of the \emph{seller}, whereas the buyer is simply myopic and participates in only one round. In contrast, it is the buyer in our model who does the learning (and there is no information for the seller to learn: the buyer's values are drawn fresh in every round). 

Finally, no-regret learning in online decision problems is an extremely well-studied problem. When feedback is revealed for every possible action, one well-known solution is the multiplicative weight update rule which has been rediscovered and applied in many fields (see survey~\cite{AroraHK12} for more details). Another algorithmic scheme for the online decision problem is known as Follow the Perturbed Leader \cite{Hannan57,KalaiV02,KalaiV05}. When only feedback for the selected action is revealed, the problem is referred to as the multi-armed bandit problem. Here, similar ideas to the MWU rule are used in developing the EXP3 algorithm~\cite{AuerCNS03} for adversarial bandit model, and also for the contextual bandit problem \cite{LangfordZ08}. Our algorithm in Theorem \ref{thm:lowregret} bears some similarities to the low swap regret algorithm introduced in \cite{BlumM07}. See the survey \cite{Bubeck12} for more details about the multi-armed bandit problem. Our results hold in both models (i.e. whether the buyer receives feedback for every bid they could have made, or only the bid they actually make), so we will make use of both classes of algorithms.

In summary, while there is already extensive work related to repeated sales in auctions, and even no-regret learning with respect to auctions (from both the buyer and seller perspective), our work is the first to address how a seller might adapt their selling strategy when faced with a no-regret buyer.

\section{Model and Preliminaries}\label{sec:prelim}
We consider a setting with 1 buyer and 1 seller. There are $T$ rounds, and in each round the seller has one item for sale. At the start of each round $t$, the buyer's value $v(t)$ (known only to the buyer) for the item is drawn independently from some distribution $\D$ (known to both the seller and the buyer). For simplicity, we assume $\D$ has a finite support\footnote{If $\D$ instead has infinite support, all our results hold approximately after discretization to multiples of $\varepsilon$. If $\D$ is bounded in $[0,H]$, then all our results hold after normalizing $\D$ by dividing by $H$.} of size $m$, supported on values $0 \leq v_1 < v_2 < \cdots < v_m \leq 1$. For each $i \in [m]$, $v_i$ has probability $q_i$ of being drawn under $\D$.

The seller then presents $K$ options for the buyer, which can be thought of as ``possible bids'' (we will interchangeably refer to these as \textit{options}, \textit{bids}, or \textit{arms} throughout the paper, depending on context). Each arm $i$ is labelled with a bid value $b_i \in [0, 1]$, with $b_1 < \ldots, < b_K$. Upon pulling this arm at round $t$, the buyer receives the item with some allocation probability $a_{i, t}$, and must pay a price $p_{i, t} \in [0, a_{i,t}\cdot b_i]$. These values $a_{i, t}$ and $p_{i, t}$ are chosen by the seller during time $t$, but remain unknown to the buyer until he plays an arm, upon which he learns the values for that arm. All of our positive results (i.e. strategies for the seller) are \emph{non-adaptive} (in some places called \emph{oblivious}), in the sense that that $a_{i,t}, p_{i,t}$ are set before the first round starts. All of our negative results (i.e. upper bounds on how much a seller can possibly attain) hold even against \emph{fully adaptive} sellers, where $a_{i,t}$ and $p_{i,t}$ can be set \emph{even after learning the distribution of arms the buyer intends to pull in round $t$}. 

In order for the selling strategies to possibly represent sponsored search auctions, we require the allocation/price rules to be monotone. That is, if $i > j$, then for all $t$, $a_{i, t} \geq a_{j, t}$ and $p_{i, t} \geq p_{j, t}$. In other words, bidding higher should result in a (weakly) higher probability of receiving the item and (weakly) higher expected payment. We'll also insist on the existence of an arm $0$ with bid $b_{0} = 0$ and $a_{0,t} = 0$ for all $t$; i.e., an arm which charges nothing but does not give the item. Playing this arm can be thought of as not participating in the auction. 

We'll be interested in one final property of allocation/price rules that we call \textbf{critical}, and buyer behavior that we call \textbf{clever}. We won't require that all auctions considered be critical, but this is an important property that greatly affects the optimal revenue that a seller can extract (see Theorems~\ref{thm:nc_seller} and~\ref{thm:mbrevubmain}). 

\begin{definition}[Clever Bidder] We say that a bidder is \emph{clever} if they never play a dominated strategy. That is, they still no-regret learn, but only over the set of bids which are not dominated. 
\end{definition}

\begin{definition}[Critical Auction] A vector of allocation/price rules (over all $t \in [T]$) is \emph{critical} if for all $t$, \emph{overbidding is a dominated strategy}.
\end{definition}

The above definition captures the property that in many auctions like GFP and GSP (both of which are critical), it makes no sense for a buyer to ever play dominated strategies - they need only learn over the undominated strategies. Note that if overbidding is \textit{strictly} dominated, any low-regret or mean-based learning algorithm will quickly learn not to overbid, and therefore play similarly to clever bidders in critical auctions. 

\subsection{Bandits and experts}

Our goal is to understand the behavior of such mechanisms when the buyer plays according to some \lowregret strategy for the multi-armed bandit problem. In the classic multi-armed bandit problem a learner (in our case, the buyer) chooses one of $K$ arms per round, over $T$ rounds. On round $t$, the learner receives a reward $r_{i,t} \in [0, 1]$ for pulling arm $i$ (where the values $r_{i,t}$ are possibly chosen adversarially). The learner's goal is to maximize his total reward.

Let $I_t$ denote the arm pulled by the principal at round $t$. The \textit{regret} of an algorithm $\mathcal{A}$ for the learner is the random variable $\Reg(\mathcal{A}) = \max_i \sum_{t=1}^T r_{i,t} - \sum_{t=1}^{T}r_{I_t, t}$. We say an algorithm $\mathcal{A}$ for the multi-armed bandit problem is \textit{$\delta$-\lowregret} if $\E[\Reg(\mathcal{A})] \leq \delta$ (where the expectation is taken over the randomness of $\mathcal{A}$). We say an algorithm $\mathcal{A}$ is \textit{\lowregret} if it is $\delta$-\lowregret for some $\delta = o(T)$. 

In the multi-armed bandits setting, the learner only learns the value $r_{i,t}$ for the arm $i$ which he pulls on round $t$. In our setting, the learner will learn $a_{i,t}$ and $p_{i,t}$ explicitly (from which they can compute $r_{i,t}$). Our results (both positive and negative) also hold when the learner learns the value $r_{i,t}$ for \emph{all} arms $i$ (we refer this full-information setting as the \textit{experts setting}, in contrast to the partial-information \textit{bandits setting}). Simple \lowregret algorithms exist in both the experts setting and the bandits setting. Of special interest in this paper will be a class of learning algorithms for the bandits problem and experts problem which we term `mean-based'. 

\begin{definition}[Mean-Based Learning Algorithm]\label{def:mb}
Let $\sigma_{i, t} = \sum_{s=1}^{t}r_{i, s}$. An algorithm for the experts problem or multi-armed bandits problem is $\gamma$-\textit{mean-based} if it is the case that whenever $\sigma_{i, t} < \sigma_{j,t} - \gamma T$, then the probability that the algorithm pulls arm $i$ on round $t$ is at most $\gamma$. We say an algorithm is \textit{mean-based} if it is $\gamma$-\textit{mean-based} for some $\gamma = o(1)$. 
\end{definition}

Intuitively, `mean-based' algorithms will rarely pick an arm whose current mean is significantly worse than the current best mean. Many \lowregret algorithms, including commonly used variants of EXP3 (for the bandits setting), the Multiplicative Weights algorithm (for the experts setting) and the Follow-the-Perturbed-Leader algorithm (experts setting), are mean-based (Appendix \ref{sect:mbalgs}). 

\subsubsection*{Contextual bandits}

In our setting, the buyer has the additional information of their current value for the item, and hence is actually facing a \textit{contextual bandits} problem. In (our variant of) the contextual bandits problem, each round $t$ the learner is additionally provided with a \textit{context} $c_t$ drawn from some distribution $\D$ supported on a finite set $C$ (in our setting, $c_t = v(t)$, the buyer's valuation for the item at time $t$). The adversary now specifies rewards $r_{i,t}(c)$, the reward the learner receives if he pulls arm $i$ on round $t$ while having context $c$. If we are in the full-information (experts) setting, the learner learns the values of $r_{i,t}(c_t)$ for all arms $i$ after round $t$, where as if we are in the partial-information (bandits) setting, the learner only learns the value of $r_{i, t}(c_t)$ for the arm $i$ that he pulled.

In the contextual bandits setting, we now define the regret of an algorithm $\mathcal{A}$ in terms of regret against the best ``context-specific'' policy $\pi$; that is, $\Reg(\mathcal{A}) = \max_{\pi:C\rightarrow[K]} \sum_{t=1}^Tr_{\pi(c_t),t}(c_t) - \sum_{t=1}^{T}r_{I_{t}, t}(c_t)$, where again $I_t$ is the arm pulled by $M$ on round $t$. As before, we say an algorithm is $\delta$-low regret if $\E[\Reg(M)] \leq \delta$, and say an algorithm is \lowregret if it is $\delta$-\lowregret for some $\delta = o(T)$. 

If the size of the context set $C$ is constant with respect to $T$, then there is a simple way to construct a \lowregret algorithm $M'$ for the contextual bandits problem from a \lowregret algorithm $M$ for the classic bandits problem: simply maintain a separate instance of $M$ for every different context $v \in C$ (in the contextual bandits literature, this is sometimes referred to as the $S$-EXP3 algorithm \cite{Bubeck12}). We call the algorithm we obtain this way its \textit{contextualization}, and denote it as $\cont(M)$. 

If we start with a mean-based learning algorithm, then we can show that its contextualization satisfies an analogue of the mean-based property for the contextual-bandits problem (proof in Appendix~\ref{sect:mbalgs}). 

\begin{definition}[Mean-Based Contextual Learning Algorithm]\label{def:mb-cont}
Let $\sigma_{i, t}(c) = \sum_{s=1}^{t}r_{i, s}(c)$. An algorithm for the contextual bandits problem is $\gamma$-\textit{mean-based} if it is the case that whenever $\sigma_{i, t}(c)< \sigma_{j,t}(c) - \gamma T$, then the probability $p_{i,t}(c)$ that the algorithm pulls arm $i$ on round $t$ if it has context $c$ satisfying $p_{i,t}(c) < \gamma$. We say an algorithm is \textit{mean-based} if it is $\gamma$-\textit{mean-based} for some $\gamma = o(1)$. 
\end{definition}

\begin{theorem}\label{thm:mean-based alg}
If an algorithm for the experts problem or multi-armed bandits problem is mean-based, then its contextualization is also a mean-based algorithm for the contextual bandits problem.
\end{theorem}

\subsection{Welfare and monopoly revenue}
In order to evaluate the performance of our mechanisms for the seller, we will compare the revenue the seller obtains to two benchmarks from the single-round setting of a seller selling a single item to a buyer with value drawn from distribution $\D$.

The first benchmark we consider is the \textit{welfare} of the buyer, the expected value the buyer assigns to the item. This quantity clearly upper bounds the expected revenue that the seller can hope to extract per round.

\begin{definition}
The \emph{welfare}, $\Val(\mathcal{D})$ is equal to $\mathbb{E}_{v\sim \mathcal{D}}[v]$.
\end{definition}

The second benchmark we consider is the \textit{monopoly revenue}, the maximum possible revenue attainable by the seller in one round against a rational buyer. Seminal work of Myerson~\cite{Myerson81} shows that this revenue is attainable by setting a fixed price (``monopoly/Myerson reserve'') for the item, and hence can be characterized as follows.

\begin{definition}
\label{def:myersonrev}
The \emph{monopoly revenue} (alternatively, \emph{Myerson revenue}) $\Mye(\mathcal{D})$ is equal to  $\max_{p}p\cdot \Pr_{v \sim \D}[v \geq p]$.
\end{definition}

\subsection{A final note on the model}
For concreteness, we chose to phrase our problem as one where a single bidder whose value is repeatedly drawn independently from $\D$ each round engages in no-regret learning with their value as context. Alternatively, we could imagine a population of $m$ different buyers, each with a \emph{fixed} value $v_i$. Each round, exactly one buyer arrives at the auction, and it is buyer $i$ with probability $q_i$. The buyers are indistinguishable to the seller, and each buyer no-regret learns (without context, because their value is always $v_i)$. This model is mathematically equivalent to ours, so all of our results hold in this model as well if the reader prefers this interpretation instead. 

\section{An Illustrative Example}
\label{sec:example}
In this section, we overview an illustrative example to show the difference between mean-based and non-mean-based learning algorithms, and between critical and arbitrary auctions. We will not prove all claims in this section (nor carry out all calculations) as it is only meant to illustrate and provide intuition. Throughout this section, the running example will be when $\D$ samples $1/4$ with probability $1/2$, $1/2$ with probability $1/4$, and $1$ with probability $1/4$. Note that $\Val(\D) = 1/2$ and $\Rev(\D) = 1/4$.

\subsection{Mean-Based Learning and Arbitrary Auctions}
\label{sec:mbaa}
Let's first consider what the seller can do with an arbitrary (not critical) auction when the buyer is running a mean-based learning algorithm like EXP3. The seller will let the buyer bid $0$ or $1$. If the buyer bids $0$, they pay nothing but do not receive the item (recall that an arm of this form is required). If the buyer bids $1$ in round $t$, they receive the item and pay some price $p_t$ as follows: for the first half of the game ($1 \leq t \leq T/2$), the seller sets $p_t = 0$. For the second half of the game ($T/2 < t \leq T$), the seller sets $p_t = 1$.

Let's examine the behaviour of the buyer, recalling that they run a mean-based learning algorithm, and therefore (almost) always pull the arm with highest cumulative utility. The buyer with value $1$ will happily bid $1$ all the way through, since he is always offered the item for less than or equal to his value for the item. The buyer with value $1/2$ will bid $1$ for the first $T/2$ rounds, accumulating a surplus (i.e., negative regret) of $1/2$ per round. For the next $T/2$ rounds, this surplus slowly disappears at the rate of $1/2$ per round until it disappears at time $T$, so the bidder with value $1/2$ will bid $1$ all the way through. Finally, the bidder with value $1/4$ will bid $1$ for the first $T/2$ rounds, accumulating surplus at a rate of $1/4$ per round. After round $T/2$, this surplus decreases at a rate of $3/4$ per round, until at round $2T/3$ his cumulative utility from bidding $1$ reaches $0$ and he switches to bidding $0$.

Now let's compute the revenue. From round $T/2$ through $2T/3$, the buyer always buys the item at a price of $1$, so the seller obtains $T/6$ revenue. Finally, from round $2T/3$ through $T$, the buyer purchases the item with probability $1/2$ and pays $1$. The total revenue is $0+T/6+T/6 = T/3$. Note that if the seller used the default strategy, they would extract revenue only $T/4$.

Where did our extra revenue come from? First, note that the welfare of the buyer in this example is quite high: the bidder gets the item the whole way through when $v \geq 1/2$, and two-thirds of the way through when $v = 1/4$. One reason why the welfare is so high is because we give the item away for free in the early rounds. But notice also that the utility of the buyer is quite low: the buyer actually has zero utility when $v \leq 1/2$, and utility $1/2$ when $v = 1$. The reason we're able to keep the utility low, despite giving the item away for free in the early rounds is because we overcharge the bidders in later rounds (and they choose to overpay, exactly because their learning is mean-based). 

In fact, by offering additional options to the buyer, we show that \emph{it is possible for the seller to extract up to the full welfare from the buyer} (e.g. a net revenue of $T/2-o(T)$ for this example). As in the above example, our mechanism makes use of arms which are initially very good for the buyer (giving the item away for free, accumulating negative regret), followed by a period where they are very bad for the buyer (where they pay more than their value). The trick in the construction is making sure that the good/bad intervals line up so that: a) the buyer purchases the item in every round, no matter their value (this is necessary in order to possibly extract full welfare) and b) by round $T$, the buyer has zero (arbitrarily small) utility, no matter their value. 

Getting the intervals to line up properly so that any mean-based learner will pick the desired arms still requires some work. But interestingly, our constructed mechanism is non-adaptive and prior-independent (i.e. the same mechanism extracts full welfare \emph{for all $\D$}). Theorem~\ref{thm:nc_seller} below formally states the guarantees. The construction itself and the proof appear in Appendix~\ref{sect:nonconservative}. 

\begin{theorem}\label{thm:nc_seller}
If the buyer is running a mean-based algorithm, for any constant $\varepsilon >0$, there exists a strategy for the seller which obtains revenue at least $(1-\varepsilon)\Val(\D)T - o(T)$. 
\end{theorem}

Two properties should jump out as key in enabling the result above. The first is that the buyer \emph{only} has no regret towards fixed arms and \emph{not} towards the policy they would have used with a lower value (this is what leads the buyer to continue bidding $1$ with value $1/2$ even though they have already learned to bid $0$ with value $1/4$). This suggests an avenue towards an improved learning algorithm: have the bidder attempt to have no regret not only towards each fixed arm, but also towards the policy of play produced when having different values. This turns out to be exactly the right idea, and is discussed in the following subsection below.

The second key property is that we were able to ``lure'' the bidders into playing an arm with a free item, then overcharge them later to make up for lost revenue. This requires that the bidder consider pulling an arm with maximum bid exceeding their value, which will never happen in a critical auction with clever bidders. It turns out it is still possible to do better than the default strategy with a critical auction against clever bidders, but not as well as with an arbitrary auction. Section~\ref{sec:minicritical} explores critical auctions for this example.

\subsection{Better Learning and Arbitrary Auctions}\label{sec:better}

In our bad example above, the buyer with value $1/2$ for the item slowly spends the second half of the game losing utility. While his behaviour is still \lowregret (he ends up with zero net utility, which indeed is at least as good as only bidding $0$), he would have been much happier to follow the actions of the buyer with value $1/4$, who started bidding $0$ at $2T/3$. 

Using this idea, we show how to construct a \lowregret algorithm for the buyer such that the seller receives at most the Myerson revenue every round. We accomplish this by extending an arbitrary \lowregret algorithm (e.g. EXP3) by introducing ``virtual arms'' for each value, so that each buyer with value $v$ has low regret not just with respect to every fixed bid, but also \lowregret with respect to the policy of play as if they had a different value $v'$ for the item (for all $v' < v$). In some ways, our construction is very similar to the construction of low internal-regret (or swap-regret) algorithms from low external-regret algorithms. The main difference is that instead of having low regret with respect to swapping actions, we have low regret with respect to swapping \textit{contexts} (i.e. values). Theorem~\ref{thm:lowregret} below states that the seller cannot outperform the default strategy against buyers who use such algorithms to learn. 

\begin{theorem}
\label{thm:lowregret}
\label{cor:lowregret}
There exists a \lowregret algorithm for the buyer against which every seller strategy extracts no more than $\Mye(\D)T + O(m\sqrt{\delta T})$ revenue.
\end{theorem}

The algorithm's description and proof appear in Appendix~\ref{sect:goodalgs}. The key observation in the proof is that ``not regretting playing as if my value were $v'$'' sounds a lot like ``not preferring to report value $v'$ instead of $v$.'' This suggests that the aggregate allocation probabilities and prices paid by any buyer using our algorithm should satisfy the same constraints as a truthful auction, proving that the resulting revenue cannot exceed the default strategy (and indeed the proof follows this approach). 

\subsection{Mean-Based Learning and Critical Auctions}\label{sec:minicritical}
Recall in our example that to extract revenue $T/3$, bidders with values $1/4$ and $1/2$ had to consider bidding $1$. If the seller is using a critical auction, overbidding is dominated, so there is no reason for bidders to do this. In fact, the analysis and results of this section hold as long as the bidders never consider overbidding (even if the auction isn't critical).

Although the auction in Section \ref{sec:mbaa} is no longer viable, consider the following auction instead: in addition to the zero arm, the bidder can bid $1/4$ or $1/2$. If they bid $1/2$ in any round, they will get the item with probability $1$ and pay $1/2$. If they bid $1/4$ in round $t \leq T/3$, they get nothing. If they bid $1/4$ in round $t \in (T/3, T]$, they get the item and pay $1/4$. Let's again see what the bidder will choose to do, remembering that they will always pull the arm that has provided highest cumulative utility (due to being mean-based).

Clearly, the bidder with value $1/4$ will bid $1/4$ every round (since they are clever, they won't even consider bidding $1/2$), making a total payment of $2T/3 \cdot 1/4 \cdot 1/2 = T/12$. The bidder with value $1/2$ will bid $1/2$ for the first $T/3$ rounds, and then immediately switch to bidding $1/4$, making a total payment of $T/3 \cdot 1/2 \cdot 1/4 + 2T/3 \cdot 1/4 \cdot 1/4 = T/12$. 

The bidder with value $1$ will actually bid $1/2$ for the entire $T$ rounds. To see this, observe that their cumulative surplus through round $t$ from bidding $1/2$ is $t \cdot 1/2 \cdot 1/4 = t/8$ ($t$ rounds by utility $1/2$ per round by probability $1/4$ of having value $1$). Their cumulative surplus through round $t$ from bidding $1/4$ is instead $(t-T/3) \cdot 3/4 \cdot 1/4 = 3t/16 - T/16 \leq t/8$ (for $t \leq T$). Because they are mean-based, they will indeed bid $1/2$ for the entire duration due to its strictly higher utility. So their total payment will be $T \cdot 1/2 \cdot 1/4 = T/8$. The total revenue is then $7T/24 > T/4$, again surpassing the default strategy (but not reaching the $T/3$ achieved by our non-critical auction). 

Let's again see where our extra revenue comes from in comparison to a truthful auction. Notice that the bidder receives the item with probability $1$ conditioned on having value $1/2$, and also conditioned on having value $1$. Yet somehow the bidder pays an average of $1/3$ conditioned on having value $1/2$, but an average of $1/2$ conditioned on having value $1$. \emph{This could never happen in a truthful auction}, as the bidder would strictly prefer to pretend their value was $1/2$ rather than $1$. But it is entirely possible when the buyer does mean-based learning, as evidenced by this example.

In Appendix \ref{sect:conservative}, we define $\Crit(\D)$ as the value of the LP in Figure \ref{fig:mblpmain}. In Theorems \ref{thm:mbrevlbmain} and \ref{thm:mbrevubmain}, we show that $\Crit(\D) T$ tightly characterizes (up to $\pm o(T)$) the optimal revenue a seller can extract with a critical auction against a clever buyer. We state the theorem statements more generally to remind the reader that they hold as long as the buyer never overbids (even if the auction is arbitrary). The proofs can be found in Appendix \ref{sec:chara}.

\begin{figure}[h]
\begin{alignat*}{2}
  \textbf{maximize }   & \sum_{i=1}^m q_i (v_i x_i - u_i)\  \\
  \textbf{subject to ~~~~} & u_i \geq (v_i-v_j) \cdot x_j,  &\ & \forall \;i,j \in [m]: i > j\\
                       & u_i \geq 0, 1 \geq x_i \geq 0,  \ &\ & \forall \; i \in [m] 
\end{alignat*}
\caption{The mean-based revenue LP.}
\label{fig:mblpmain}
\end{figure}

Before stating our theorems, let's parse this LP. $q_i$ is a constant representing the probability that the buyer has value $v_i$ (also a constant). $x_i$ is a variable representing the average probability that the bidder gets the item with value $v_i$, and $u_i$ is a variable representing the average utility of the bidder when having value $v_i$. Therefore, this bidder's average value is $v_i x_i$, the average price they pay is $v_i x_i - u_i$, and the objective function is simply the average revenue. The second constraints are just normalization, ensuring that everything lies in $[0,1]$. The first line of constraints are the interesting ones. These look a lot like IC constraints that a truthful auction must satisfy, but something's missing: the LHS is clearly the utility of the buyer with value $v_i$ for ``telling the truth,'' but the utility of the buyer for ``reporting $v_j$ instead'' is $(v_i - v_j)\cdot x_j + u_j$. So the $u_j$ term is missing on the RHS. 

Let's also see a very brief proof outline for why no seller can extract more revenue than $\Crit(\D)$: 
\begin{enumerate}
\item Because the buyer has no regret conditioned on having value $v_i$, their utility is at least as high as playing arm $j$ every round. 
\item Because the auction never charges arm $j$ more than $v_j$ (conditioned on awarding the item), the buyer's utility for playing arm $j$ every round is at least $y_j\cdot (v_i - v_j)$, where $y_j$ is the average probability that arm $j$ awards the item.
\item Because the auction is monotone, and the buyer never considers overbidding, if the buyer gets the item with probability $x_j$ conditioned on having value $v_j$, we must have $y_j \geq x_j$. 
\end{enumerate}

These three facts together show that no seller can extract more than $\Crit(\D)$ against a no-regret buyer who doesn't overbid. Observe also that step 3 is \emph{exactly} the step that doesn't hold for buyers who consider overbidding (and is exactly what's violated in our example in Section~\ref{sec:mbaa}): if the buyer ever overbids, then they might receive the item with higher probability than had they just played their own arm every round. 

\begin{theorem}
\label{thm:mbrevubmain}
Any strategy for the seller achieves revenue at most $\Crit(\D)T + o(T)$ against a buyer running a \lowregret algorithm who overbids with probability $0$. 
\end{theorem}  

\begin{theorem}
\label{thm:mbrevlbmain}
For any constant $\varepsilon > 0$, there exists a strategy for the seller gets revenue at least $(\Crit(\D) - \eps)T - o(T)$ against a buyer running a mean-based algorithm who overbids with probability $0$. The strategy sets a decreasing cutoff $r_t$ and for all $t$ awards the item with probability $1$ to any bid $b_t \geq r_t$ for price $b_t$, and with probability $0$ to any bid $b_t < r_t$.
\end{theorem}

\begin{theorem}
\label{thm:mbrev}
For distributions $\D$ supported on $[1/H, 1]$, $\Crit(\D) = O(\log \log H)$, and there exist $\D$ supported on $[1/H, 1]$ such that $\Crit(\D) = \Theta(\log \log H)$. For this same $\D$, $\Crit(\D) = \Theta(\log H)$.\footnote{The promised $\D$ is the equal-revenue curve truncated at $H$.}
\end{theorem}

\subsection{A Final Note on the Example}
While reading through our examples, the reader may think that the mean-based learner's behavior is clearly irrational: why would you continue paying above your value? Why would you continue paying more than necessary, when you can safely get the item for less?

But this is exactly the point: a more thoughtful learner can indeed do better (for instance, by using the algorithm of Section~\ref{sec:better}). It is also perhaps misleading to believe that the bidder should ``obviously'' stop overpaying: we only know this because we know the structure of the example. But in principle, how is the bidder supposed to know that the overcharged rounds are the new norm and not an anomaly? Given that most standard no-regret algorithms are mean-based, it's important to nail down the seller's options for exploiting this behavior.

\section{Conclusion and Future Directions}
\label{sec:conclusions}
Motivated by the prevalence of bidders no-regret learning to play non-truthful auctions in practice~\cite{NekipelovST15}, we consider a revenue-maximizing seller with a single item (each round) to sell to a single buyer. We show that when the buyer uses mean-based algorithms like EXP3, the seller can extract revenue equal to the expected welfare with an unnatural auction. We then provide a modified no-regret algorithm $\mathcal{A}$ such that the seller cannot extract revenue exceeding the monopoly revenue when the buyer bids according to $\mathcal{A}$. Finally, we consider a mean-based buyer who never overbids. We tightly characterize the seller's optimal revenue with a linear program, and show that a pay-your-bid auction with decreasing reserves over time achieves this guarantee. Moreover, we show that the mean-based revenue can be unboundedly better than the monopoly revenue while simultaneously worse than the expected welfare. In particular, for the equal revenue curve truncated at $H$, the monopoly revenue is $1$, the expected welfare is $\ln(H)$, and the mean-based revenue is $\Theta(\ln (\ln(H)))$. 

While our work has already shown the single-buyer problem is quite interesting, the most natural direction for future work is understanding revenue maximization with multiple learning buyers. Of our three main results, only Theorem~\ref{thm:lowregret} extends easily (that if every buyer uses our modified learning, the default strategy, which now runs Myerson's optimal auction every round, is optimal; see Theorem \ref{thm:alg_buyer_multi} for details). Our work certainly provides good insight into the multi-bidder problem, but there are still clear barriers. For example, in order to obtain revenue equal to the expected welfare, the auction must necessarily also maximize welfare. In our single-bidder model, this means that we can give away the item for free for $\Omega(T)$ rounds, but with multiple bidders, such careless behaviour would immediately make it impossible to achieve the optimal welfare. Regarding the mean-based revenue, while there is a natural generalization of our LP to multiple bidders, it's no longer clear how to achieve this revenue with a critical auction, as all the relevant variables now implicitly depend on the actions of the other bidders. These are just examples of concrete barriers, and there are likely interesting conceptual barriers for this extension as well.

Another interesting direction is understanding the consequences of our work from the perspective of the buyer. Aside from certain corner configurations (e.g. the seller extracting the buyer's full welfare), it's not obvious how the buyer's utility changes. For instance, is it possible that the buyer's utility actually \emph{increases} as the seller switches from the default strategy to the optimal mean-based revenue? Does the buyer ever benefit from using an ``exploitable'' learning strategy, so that the seller can exploit it and make them both happier?

\bibliographystyle{alpha}
\bibliography{bib}

\appendix

\section{Good \lowregret algorithms for the buyer}\label{sect:goodalgs}

In this section we show that there exists a (contextual) \lowregret algorithm for the buyer which guarantees that the seller receives at most the Myerson revenue per round (i.e., $\Mye(\D)T$ in total). As mentioned earlier, it does not suffice for the buyer to simply run the contextualization $\cont(M)$ for some \lowregret learning algorithm $M$ (and in fact, if $M$ is mean-based, the seller can extract strictly more than $\Mye(\D)T$, as we will see later). However, by modifying $\cont(M)$ so that it has not just \lowregret with respect to the best stationary policy, but so that it additionally does not regret playing as if it had some other context, we obtain a \lowregret algorithm for the buyer which guarantees the seller receives no more than $\Mye(\D)$ per round.

The details of the algorithm are presented in Algorithm \ref{alg:lowreg}. Recall that the distribution $\D$ is supported over $m$ values $v_1 < v_2 < \cdots < v_m$, where for each $i \in [m]$, $v_i$ has probability $q_i$ under $\D$. The algorithm takes a \lowregret algorithm $M$ for the classic multi-armed bandit problem, and runs $M$ instances of it, one per possible value $u$. Each instance $M_i$ of $M$ learns not only over the possible $K$ actions, but also over $i-1$ virtual actions corresponding to values $v_1$ through $v_{i-1}$. Picking the virtual action associated with $v_j$ corresponds to the buyer pretending they have value $v_j$, and playing accordingly (i.e., querying $M_j$). 

This algorithm is very similar in structure to the construction of a low swap-regret bandits algorithm from a generic \lowregret bandits algorithm (see \cite{BlumM07}). The main difference is that whereas swap regret guarantees \lowregret with respect to swapping actions (i.e. always playing action $i$ instead of action $j$), this algorithm guarantees \lowregret with respect to swapping \textit{contexts} (i.e., always pretending you have context $i$ when you actually have context $j$). In addition, the auction structure of our problem allows us to only consider contexts with valuations smaller than our current valuation $v_i$; this puts a limit of $m$ on the number of recursive calls per round, as opposed to the low swap regret algorithm where one must solve for the stationary distribution of a Markov chain over $m$ states each round.

\begin{algorithm}[ht] 
	\caption{No-regret algorithm for buyer.}\label{alg:lowreg}
    \begin{algorithmic}[1]
    	\STATE Let $M$ be a $\delta$-\lowregret algorithm for the classic multi-armed bandit problem, with $\delta = o(T)$. Initialize $m$ copies of $M$, $M_1$ through $M_m$.
        \STATE Instance $M_i$ of $M$ will learn over $K+i-1$ arms. 
        \STATE The first $K$ arms of $M_i$ (``bid arms'') correspond to the $K$ possible menu options $b_1, b_2, \dots, b_K$. 
        \STATE The last $i-1$ arms of $M_i$ (``value arms'') correspond to the $i-1$ possible values (contexts) $v_1,\dots,v_{i-1}$. 
        \FOR {$t=1$ to $T$}
        	\IF {buyer has value $v_i$}
        		\STATE Use $M_i$ to pick one arm from the $K+i-1$ arms.
		\IF {the arm is a bid arm $b_j$}
			\STATE Pick the menu option $j$ (i.e. bid $b_j$).
		\ELSIF {the arm is a value arm $v_j$}
			\STATE Sample an arm from $M_j$ (but don't update its state). If it is a bid arm, pick the corresponding menu option. If it is a value arm, recurse.
		\ENDIF
		\STATE Update the state of algorithm $M_i$ with the utility of this round.
        	\ENDIF
        \ENDFOR
      \end{algorithmic}
\end{algorithm}

We now proceed to show that Algorithm \ref{alg:lowreg} has our desired guarantees.

\begin{theorem}
\label{thm:alg_buyer}
Let $q_{min} = \min_{i} q_i$. If the buyer plays according to Algorithm \ref{alg:lowreg} then the seller (even if they play an adaptive strategy)  receives no more than $\Mye(\D) T + \frac{m\delta}{q_{min}}$ revenue. 
\end{theorem}

\begin{proof}
For each $i \in [m]$, define $h_i$ to be the expected number of rounds the buyer receives the item when they have value $v_i$. For each $i \in [m]$ define $r_i$ to be the expected total payment from the buyer to the seller when the buyer has value $v_i$. Our goal is to upper bound $\sum_{i} r_i$, the total revenue the seller receives.

Recall that every strategy must contain a zero option in its menu, where the buyer pays nothing and doesn't receive the item (and hence receives zero utility). Since each $M_i$ is a $\delta$-\lowregret algorithm, we know that the buyer does not regret always choosing the zero option when they have value $v_i$. It follows that, for all $i \in [m]$, we have that

\begin{equation}\label{eqn:lowregbid}
v_i h_i - r_i \geq -\delta.
\end{equation}

The following lemma shows that when $j>i$, the buyer does not regret pretending to have value $v_i$ when they have value $v_j$. 

\begin{lemma}\label{lem:lowregval}
For all $1 \leq i < j \leq m$,
\[
(v_j h_j - r_j)/q_j \geq (v_j h_i - r_i)/q_i  - \delta/q_j.
\]
\end{lemma}

\begin{proof}
From the algorithm, we know that $M_j$ does not regret always playing the value arm corresponding to $v_i$. We define the following notation. For all $i \in [m], t \in [T]$ and any history $\pi$ of $t-1$ rounds (including for each round which option is chosen and the utility of that round), define $h_i (t,\pi)$ to be the probability of getting item in round $t$ given history $\pi$ when buyer has value $v_i$ and define $r_i(t,\pi)$ to be the expected price paid in round $t$ when the buyer has value $v_i$ given history $\pi$. 

Let $\Pi_t$ be the distribution of histories at round $t$, for $t = 0,...,T-1$. The \lowregret guarantee tells us that 

\begin{equation}\label{eqn:lowregval}
\sum_{t=1}^T q_j\cdot \mathbb{E}_{\pi \sim \Pi_{t-1}} \left[(h_j(t,\pi) v_j - r_j(t,\pi)) - (h_i(t,\pi)v_j - r_i(t,\pi))\right] \geq -\delta .
\end{equation}

Note that 

\begin{eqnarray*}
\sum_{t=1}^T \mathbb{E}_{\pi \sim \Pi_{t-1}} [h_j(t,\pi)q_j]  &=& h_j, \\
\sum_{t=1}^T \mathbb{E}_{\pi \sim \Pi_{t-1}} [h_i(t,\pi)q_i] &=& h_i, \\
\sum_{t=1}^T \mathbb{E}_{\pi \sim \Pi_{t-1}} [r_j(t,\pi)q_j]  &=& r_j, \\
\sum_{t=1}^T \mathbb{E}_{\pi \sim \Pi_{t-1}} [r_i(t,\pi)q_i] &=& r_i.
\end{eqnarray*}

Dividing (\ref{eqn:lowregval}) through by $q_j $ and substituting in these relations, we arrive at the statement of the lemma.
\end{proof}

Now define $\lambda_i = \sum_{j \leq i}\frac{1}{q_j}$, and define
\begin{equation}\label{eqn:ridef}
r'_i = \frac{r_i}{q_i} - \lambda_i\delta.
\end{equation}

It follows from Lemma \ref{lem:lowregval} that for all $1 \leq i < j \leq m$,
\begin{equation}\label{eqn:rival}
\frac{v_j h_j}{q_j} - r'_j \geq \frac{v_j h_i}{q_i} - r'_i.
\end{equation}
From (\ref{eqn:lowregbid}), we also have for all $i \in [m]$,

\begin{equation}\label{eqn:ribid}
\frac{v_i h_i}{q_j} -r'_i \geq 0.
\end{equation}

We will argue from these constraints that $\sum_{i}q_ir'_i \leq \Mye(\D)T$. To do this, we will construct a single-round mechanism for selling an item to a buyer with value distribution $\D$ such that this mechanism has expected revenue $\sum_{i}q_ir'_i/T$; the result then follows from the optimality of the Myerson mechanism (\cite{Myerson81}). 

To construct this mechanism, first find a sequence of indices $a_1,a_2,\dots,a_l$ via the following algorithm.
\begin{algorithm}[ht]
    \begin{algorithmic}[1]
        \STATE $l\leftarrow 1$, $a_1 \leftarrow 1$. 
        \FOR { $i = 2$ to $m$}
        		\IF {$r'_{a_i} \geq r'_{a_l}$}
			\STATE $l \leftarrow l+1$, $a_l \leftarrow i$.
		\ENDIF
        \ENDFOR
             \end{algorithmic}
\end{algorithm}

It is easy to verify that following this algorithm results in $r'_{a_1} \leq r'_{a_2} \leq \cdots \leq r'_{a_l}$. For any $a_i \leq j < a_{i+1}$ (assuming $a_{l+1} = m+1$), $r'_j < r'_{a_i}$. 
\begin{lemma}
For a bidder with value distribution $\D$, the following menu of $l$ options will achieve revenue at least $\sum_{i=1}^m r'_i q_i/T$: for each $1 \leq i \leq l$, the buyer has the choice of paying $r'_{a_i}/T$, and receiving the item with probability $h_{a_i}/(q_{a_i}T)$.
\end{lemma}
\begin{proof}
Consider some value $v_j$ in $\D$. We will show that the buyer with value $v_j$ will pay at least $r'_j/T$, thus proving the lemma. Assume $a_i \leq j \leq a_{i+1}$. 

We have (from (\ref{eqn:ribid}) and the monotonicity of $v_i$) that
\[
\frac{v_j h_{a_i}}{q_{a_i}} - r'_{a_i} \geq \frac{v_{a_i} h_{a_i}}{q_{a_i}}-r'_{a_i} \geq 0.
\]
This means the buyer with value $u_j$ receives non-negative utility by choosing option $i$. 
For any $1\leq i' < i$,  we have (from (\ref{eqn:rival})) that
\[
\frac{v_{a_i} h_{a_i}}{q_{a_i}} - r'_{a_i} \geq \frac{v_{a_i} h_{a_{i'}}}{q_{a_{i'}}} - r'_{a_{i'}}.
\]
Since $r'_{a_i} \geq r'_{a_{i'}}$,  the above inequality implies that
\[
\frac{h_{a_i}}{q_{a_i}} \geq \frac{h_{a_{i'}}}{q_{a_{i'}}}.
\]
It follows that
\[
v_j \left(\frac{h_{a_i}}{q_{a_i}} - \frac{h_{a_{i'}}}{q_{a_{i'}}}\right) \geq v_{a_i}\left(\frac{h_{a_i}}{q_{a_i}} - \frac{h_{a_{i'}}}{q_{a_{i'}}}\right)  \geq r'_{a_i} - r'_{a_{i'}}.
\]
This means the buyer with value $v_j$ prefers option $i$ to all options $i' < i$. Therefore this buyer will choose an option from $\{i, i+1, \dots, l\}$. Since $r'_{j} \leq r'_{a_i} \leq r'_{a_{i+1}} \leq \dots \leq r'_{a_l}$, we know that this buyer will pay at least $r'_j/T$, as desired.
\end{proof}

It follows from the optimality of the Myerson auction that $\sum_{i}q_ir'_i/T \leq \Mye(\D)$, and therefore that $\sum_{i}q_ir'_i \leq \Mye(\D)T$. Expanding out $r'_i$ via (\ref{eqn:ridef}), we have that 

\begin{eqnarray*}
\sum_{i}q_ir'_i &=& \sum_{i}r_i - \sum_{i}q_i\lambda_{i}\delta \\
&\geq & \sum_{i}r_i - \delta \cdot \max_{i}\lambda_{i} \\
&\geq & \sum_{i}r_i - \frac{m\delta}{q_{min}}, \\
\end{eqnarray*}

\noindent
from which the theorem follows.
\end{proof}

We can remove the explicit dependence on $q_{min}$ by filtering out all values which occur with small enough probability.

\begin{corollary}[Restatement of Theorem \ref{thm:lowregret}]
\label{cor:lowregret}
There exists a \lowregret algorithm for the buyer where the seller receives no more than $\Mye(\D)T + O(m\sqrt{\delta T})$ revenue.
\end{corollary}
\begin{proof}
Ignore all values $v_i$ with $q_i \leq \sqrt{\delta/T}$ (whenever a round with this value arises, choose an arbitrary action for this round). There are $m$ total values, so this happens with at most probability $m\sqrt{\delta/T}$, and therefore modifies the regret and revenue in expectation by at most $O(m\sqrt{\delta T}) = o(T)$. 

The regret bound from Theorem \ref{thm:alg_buyer} then holds with $q_{min} \geq \sqrt{\delta/T}$, from which the result follows.
\end{proof}

\subsection{Multiple bidders}

Interestingly, we show that by slightly modifying Algorithm \ref{alg:lowreg}, we obtain an algorithm (Algorithm \ref{alg:lowregmulti}) that works for the case where there are \textit{multiple bidders}. In the multiple bidder setting, there are $B$ bidders with independent valuations for the item. Each round $t$, bidder $\ell$ receives a value $v_{\ell}(t)$ for the item drawn from a distribution $\D_{\ell}$ (independently of all other values). Each distribution $\D_{\ell}$ is supported over $m_{\ell}$ values, $v_{\ell, 1} < v_{\ell, 2} < \dots < v_{\ell, m_{\ell}}$, where $v_{\ell, i}$ occurs under $\D_{\ell}$ with probability $q_{\ell, i}$. Every round each bidder $\ell$ submits a bid $b_{\ell}(t)$, and the auctioneer decides on an allocation rule $\mathbf{a}_{t}$, which maps $\ell$-tuples of bids $(b_1(t), b_2(t), \dots, b_{B}(t))$ to $\ell$-tuples of probabilities $(a_1(t), a_2(t), \dots, a_{B}(t))$ and a pricing rule $\mathbf{p}_{t}$, which maps $\ell$-tuples of bids $(b_1(t), b_2(t), \dots, b_{B}(t))$ to $\ell$-tuples of prices $(p_1(t), p_2(t), \dots, p_{B}(t))$. The allocation rule $\mathbf{a}_{t}$ must additionally obey the supply constraint that $\sum_{\ell} a_{\ell}(t) \leq 1$. Bidder $\ell$ wins the item with probability $a_{\ell}(t)$ and pays $p_{\ell}(t)$.

We show that if every bidder plays the \lowregret algorithm Algorithm \ref{alg:lowregmulti}, then the auctioneer (even if playing adaptively) is guaranteed to receive no more than $\Mye(\D_1, \D_2, \dots, \D_B)T + o(T)$ revenue, where $\Mye(\D_1, \D_2, \dots, \D_B)$ is the optimal revenue obtainable by an auctioneer selling a single item to $B$ bidders with valuations drawn independently from distributions $\D_{\ell}$. In other words, if every bidder plays according to Algorithm \ref{alg:lowregmulti}, the seller can do nothing better than running the single-round optimal Myerson auction every round.

The only difference between Algorithm \ref{alg:lowreg} and Algorithm \ref{alg:lowregmulti} is that instance $M_i$ in Algorithm \ref{alg:lowregmulti} has a value arm for every possible value, not only the values less than $v_i$. This means that the recursion depth of this algorithm is potentially unlimited, however it will still terminate in finite expected time since we insist that $M$ has a positive probability of picking any arm (in particular, it will eventually pick a bid arm). We can optimize the runtime of step 11 of Algorithm \ref{alg:lowregmulti} by eliciting a probability distribution over arms from each instance $M_i$, constructing a Markov chain, and solving for the stationary distribution. This takes $O((K+m)^3)$ time per step of this algorithm.

\begin{algorithm}[ht] 
	\caption{No-regret algorithm for a bidder (when there are multiple bidders).}\label{alg:lowregmulti}
    \begin{algorithmic}[1]
    	\STATE Let $M$ be a $\delta$-\lowregret algorithm for the classic multi-armed bandit problem (that always has some positive probability of choosing any arm), with $\delta = o(T)$. Initialize $m$ copies of $M$, $M_1$ through $M_m$.
        \STATE Instance $M_i$ of $M$ will learn over $K+m$ arms. 
\STATE The first $K$ arms of $M_i$ (``bid arms'') correspond to the $K$ possible menu options $b_1, \dots, b_K$.
\STATE The last $m$ arms of $M_i$ (``value arms'') correspond to the $m$ possible values (contexts) $v_1,\dots,v_{m}$.
        \FOR {$t=1$ to $T$}
        	\IF {buyer has value $v_i$}
        		\STATE Use $M_i$ to pick one arm from the $K+m$ arms.
		\IF {the arm is a bid arm $b_j$}
			\STATE Pick the menu option $j$ (i.e. bid $b_j$).
		\ELSIF {the arm is a value arm $v_j$}
			\STATE Sample an arm from $M_j$ (but don't update its state). If it is a bid arm, pick the corresponding menu option. If it is a value arm, recurse.
		\ENDIF
		\STATE Update the state of algorithm $M_i$ with the utility of this round.
        	\ENDIF
        \ENDFOR
      \end{algorithmic}
\end{algorithm}

\begin{theorem}
\label{thm:alg_buyer_multi}
Let $q_{min} = \min_{\ell, i} q_{\ell, i}$. If every bidder plays according to Algorithm \ref{alg:lowregmulti} then the auctioneer (even if they play an adaptive strategy)  receives no more than $\Mye(\D_1, \D_2, \dots, \D_B) T + O\left(\sqrt{\frac{\delta T}{q_{min}}}\right)$ revenue. 
\end{theorem}
\begin{proof}
Similarly as before, let $h_{\ell, i}$ equal the expected number of rounds bidder $\ell$ receives the item while having value $v_{\ell, i}$, and let $r_{\ell, i}$ equal the expected total amount bidder $\ell$ pays to the auctioneer while having value $v_{\ell, i}$. Again, our goal is to upper bound $\sum_{\ell}\sum_{i} r_{\ell, i}$, the total expected revenue the seller receives. 

Note that, as before, since every strategy contains a zero option in its menu, we have that (for all $\ell \in [B]$ and $i \in [m_{\ell}]$)

\begin{equation}\label{eqn:multi_ir}
v_{\ell, i}h_{\ell,i} - r_{\ell, i} \geq -\delta.
\end{equation}
Repeating the argument of Lemma \ref{lem:lowregval} (which still holds in the multiple bidder setting), we additionally have that (for all $\ell \in [B]$ and $1 \leq i < j \leq m_{\ell}$), 

\begin{equation}\label{eqn:multi_bic}
\frac{v_{\ell, j}h_{\ell, j} - r_{\ell, j}}{q_{\ell, j}} \geq \frac{v_{\ell, j}h_{\ell, i} - r_{\ell, i}}{q_{\ell, i}} - \frac{\delta}{q_{\ell, j}}.
\end{equation}

We will now (as in the proof of Theorem \ref{thm:alg_buyer}) construct a mechanism for the single-round instance of the problem of an auctioneer selling a single item to $B$ bidders with valuations independently drawn from $\D_{\ell}$. Our mechanism $M$ will work as follows:

\begin{enumerate}
\item The auctioneer will begin by asking each of the bidders for their valuations. Assume that bidder $\ell$ reports valuation $v'_{\ell}$ (we will insist that $v'_{\ell}$ belongs to the support of $\D_{\ell}$). 
\item The auctioneer will then sample a $t \in [T]$ uniformly at random.
\item For each bidder $\ell$, the auctioneer will calculate $a_{\ell}(t)$ and $p_{\ell}(t)$, the expected allocation probability and price bidder $\ell$ has to pay in round $t$ of the dynamic $T$-round mechanism, \textit{conditioned on $v_{\ell}(t) = v'_{\ell}$ for all $\ell$}.
\item The auctioneer will then give the item to bidder $\ell$ with probability $a_{\ell}(t)$, and charge bidder $\ell$ a price $p_{\ell}(t)$.
\end{enumerate}

Note that since the allocation rules $\mathbf{a}_t$ must always satisfy the supply constraint, the probabilities $a_{\ell}(t)$ we sample also obey this supply constraint, and therefore this is a valid mechanism for the single-round problem. We will now show it is approximately incentive compatible.

\begin{lemma}\label{lem:eps-bic}
Mechanism $M$ is $\frac{\delta}{q_{min}T}$-Bayesian incentive compatible and $\frac{\delta}{q_{min}T}$-ex-interim individually rational.
\end{lemma}
\begin{proof}
To begin, we claim that in expectation, if bidder $\ell$ reports valuation $v_{\ell, i}$ (and everyone else reports truthfully), then the expected probability bidder $\ell$ receives the item (under this single-round mechanism) is equal to $h_{\ell, i}/Tq_{\ell, i}$. Likewise, we claim that, if bidder $\ell$ reports valuation $v_{\ell, i}$ (and everyone else reports truthfully), the expected payment bidder they pay is equal to $r_{\ell, i}/Tq_{\ell, i}$. 

To see why this is true, let $h_{\ell}(t, i_1, i_2, \dots, i_B)$ equal the probability bidder $\ell$ gets the item (in the multi-round mechanism) at time $t$ conditioned on $v_{\ell}(t) = v_{\ell, i}$ for all $\ell \in [B]$. By construction, the probability $a'_{\ell, i}$ bidder $\ell$ receives the item (in mechanism $M$) after reporting valuation $v_{\ell, i}$ is equal to

$$a'_{\ell, i} = \frac{1}{T}\sum_{t}\sum_{\ell' \neq \ell, v_{\ell', i_{\ell'}} \in \supp \D_{\ell'}} \prod_{\ell' \neq \ell} q_{\ell', i_{\ell'}} h_{\ell}(t, i_1, i_2, \dots, i_{\ell-1}, i, i_{\ell+1}, \dots, i_B).$$

On the other hand, we can write $h_{\ell, i}$ in terms of our function $h_{\ell}$ a

$$h_{\ell, i} = \sum_{t}\sum_{\ell' \neq \ell, v_{\ell', i_{\ell'}} \in \supp \D_{\ell'}} q_{\ell, i}\prod_{\ell' \neq \ell} q_{\ell', i_{\ell'}} h_{\ell}(t, i_1, i_2, \dots, i_{\ell-1}, i, i_{\ell+1}, \dots, i_B).$$

It follows that $a'_{\ell, i} = \frac{h_{\ell,i}}{Tq_{\ell, i}}$. A similar calculation shows that if $p'_{\ell, i}$ is the expected payment of bidder $\ell$ (if they report valuation $v_{\ell, i}$ and everyone else reports truthfully), then $p'_{\ell, i} = \frac{r_{\ell, i}}{Tq_{\ell, i}}$.

Now, recall that a mechanism is $\epsilon$-BIC if misreporting your value increases your expected utility by at most $\epsilon$ (assuming everyone else reports truthfully). To show that mechanism $M$ is $\epsilon$-BIC, it therefore suffices to show that for all $j \neq i$, that

$$a'_{\ell, j}v_{\ell, i} - p'_{\ell, j} \leq a'_{\ell, i}v_{\ell, i} - p'_{\ell, i} + \epsilon.$$

But for $\epsilon = \delta/(q_{min}T)$, this follows from equation (\ref{eqn:multi_bic}). Similarly, $M$ is $\epsilon$-ex-interim IR if for all $i$, 

$$a'_{\ell, i}v_{\ell, i} - p'_{\ell, i} \geq -\epsilon.$$

Again, this follows from equation (\ref{eqn:multi_ir}), and the result therefore follows.
\end{proof}

We now apply the following lemma from \cite{DaskalakisW12}, which lets us transform an $\epsilon$-BIC mechanism $M$ into a BIC mechanism $M'$ at the cost of $O(\sqrt{\eps})$ revenue.

\begin{lemma}\label{lem:dw}
If $M$ is an $\epsilon$-BIC, $\epsilon$-ex-interim IR mechanism for selling a single item to several bidders with independent valuations, then there exists a BIC, ex-interim IR mechanism $M'$ for the same problem that satisfies $\Rev(M') \geq \Rev(M) - O(\sqrt{\epsilon})$.
\end{lemma}
\begin{proof}
See Theorem 3.3 in \cite{DaskalakisW12}.
\end{proof}

Applying Lemma \ref{lem:dw} to our mechanism, we obtain a mechanism $M'$ that satisfies $\Rev(M') \geq \Rev(M) - O(\sqrt{\frac{\delta}{q_{min} T}})$. Finally, note that since the Myerson auction is the optimal Bayesian-incentive compatible mechanism for this problem, $\Rev(M') \leq \Mye(\D_1, \dots, \D_{B})$. On the other hand, since (from the proof of Lemma \ref{lem:eps-bic}) the expected payment bidder $\ell$ pays under mechanism $M$ when being truthful is equal to:

$$\sum_{i} q_{\ell, i} \cdot \frac{r_{\ell, i}}{Tq_{\ell, i}} = \frac{1}{T} \sum_{i} r_{\ell, i}.$$

It follows that

$$\frac{1}{T}\sum_{\ell}\sum_{i}r_{\ell, i} \leq \Mye(\D_1, \dots, \D_{B}) + O(\sqrt{\frac{\delta}{q_{min} T}}),$$

and thus that

$$\sum_{\ell}\sum_{i}r_{\ell, i} \leq \Mye(\D_1, \dots, \D_{B})T + O(\sqrt{\frac{\delta T}{q_{min}}}).$$

\end{proof}

\section{Achieving full welfare against non-conservative buyers}\label{sect:nonconservative}

In this section, we will show that if the buyer uses a mean-based algorithm instead of Algorithm \ref{alg:lowreg}, the seller has a strategy which extracts the entire welfare from the buyer (hence leaving the buyer with zero utility).

\begin{theorem}[Restatement of Theorem \ref{thm:nc_seller}]
If the buyer is non-conservative and running a mean-based algorithm, for any constant $\varepsilon >0$, there exists a strategy for the seller which obtains revenue at least $(1-\varepsilon)\Val(\D)T - o(T)$. 
\end{theorem}

\begin{proof}
If every element in the support of $\D$ is at least $1- \varepsilon$, then the seller can simply always sell the item at price $1 - \varepsilon$ (since $\D$ is supported on $[0,1]$, this ensures a $(1-\eps)$ approximation to the buyer's welfare). From now on, we will assume that $\D$ is not entirely supported on $[1-\eps, 1]$.

Recall that $\D$ is supported on $m$ values $v_1 < v_2 < \dots < v_m$, where $v_i$ is chosen with probability $q_i$. Define $\rho = \min (v_m, 1 - \varepsilon/2)$, and define $\delta = (1-\rho)/(1-v_1)$. Since $v_1 < 1-\varepsilon/2$ and $v_1 < v_m$, we know that $v_1 < \rho$ and therefore $\delta < 1$. Notice that here we can make the strategy independent of $\D$ if we just pick $\rho = 1 - \varepsilon/2$ and $\delta = \varepsilon / 2$ (but setting $\rho$ and $\delta$ according to information about $\D$ can reduce the number of arms). 

Consider the following strategy for the seller. In addition to the zero arm, the seller will offer $n = \frac{\log (\eps/2)}{\log (1-\delta)}$ possible options, each with maximum bid value $b_i = 1$. We divide the timeline of each arm into three ``sessions'' in the following way:

\begin{enumerate}
\item \textbf{$\emptyset$ session:} For the first $(1- (1-\delta)^{i-1})T$ rounds, the seller charges 0 and does not give the item to the buyer (i.e. $(p_{i,t}, q_{i,t}) = (0,0)$).
\item \textbf{0 session:} For the next $(1-\delta)^{i-1} (1-\rho) T$ rounds, the seller charges 0 and gives the item to the buyer (i.e. $(p_{i,t}, q_{i,t}) = (0,1)$).
\item \textbf{1 session:} For the final $(1-\delta)^{i-1} \rho T$ rounds, the seller charges 1 and gives the item to the buyer (i.e. $(p_{i,t}, q_{i,t}) = (1,1)$).
\end{enumerate}

Note that this strategy is monotone; if $i > j$, then $p_{i,t} \geq p_{j, t}$ and $a_{i,t} \geq a_{j,t}$. 
 
Assume that the buyer is running a $\gamma$-mean-based algorithm, for some $\gamma = o(1)$. Define $A_{j} = (1 - \rho(1-\delta)^{j-1})T$ and $B_{j}(v) = A_{j} + \frac{\min(v, \rho)}{1-v_1}(1-\rho)(1-\delta)^{j-1}T - \gamma T$. Note that $A_j$ is the round where arm $j$ starts its $1$ session; we show in the following Lemma that (by the mean-based property), the buyer with value $v$ will prefer arm $j$ over any arm $j' < j$ over all rounds in the interval $[A_j, B_j(v)]$. 

 \begin{lemma}
 For each $v_i \in \D$, $j \in \{1, \dots, n-1\}$, and round $\tau \in [A_j, B_j(v_i)]$, $\sigma_{j, \tau}(v_i) > \sigma_{j', \tau}(v_i) + \gamma T$ for all $j' > j$. 
 \end{lemma}
 
 \begin{proof}
Note that arm $j$ starts its 1 session at round $A_{j} \leq \tau$. It follows that 

\begin{eqnarray*}
\sigma_{j,\tau}(v_i) &=& v_i\left(\tau-(1-(1-\delta)^{j-1})T\right) - \left((1-\delta)^{j-1}T\rho-(T- \tau)\right) \\
&=& (T - \tau) + (v_i - \rho)(1-\delta)^{j-1} T+ v_i \tau - Tv_i.
\end{eqnarray*}

Now consider the cumulative utility of playing some arm $j' > j$. It is easy to verify that $B_{j} < A_{j+1}$, and therefore arm $j'$ is still either in its $\emptyset$ session or its $0$ session. Since arm $j+1$ starts its $0$ session the earliest, it follows that $\sigma_{j', \tau}(v_i) \leq \sigma_{j+1, \tau}(v_i)$, so from now on, assume without loss of generality that $j'=j+1$. There are two cases:

 \begin{enumerate}
 \item If $\tau < T(1- (1-\delta)^j)$, the utility is 0.
 \item If $\tau \geq T(1-(1-\delta)^j)$, the utility is $(\tau - T(1-(1-\delta)^j))v_i$.
 \end{enumerate}
 
It suffices to show that 
$$(T - \tau) + (v_i - \rho)(1-\delta)^{j-1} T+ v_i \tau - Tv_i \geq \max (0 , (\tau - T(1-(1-\delta)^j))v_i) + \gamma T.$$ 

We have that
 \begin{eqnarray*}
&& (T - \tau) + (v_i - \rho)(1-\delta)^{j-1} T+ v_i \tau - Tv_i - (\tau - T(1-(1-\delta)^j))v_i \\
&=& v_i (1-\delta)^{j-1} \delta T  + (T - \tau) -\rho(1-\delta)^{j-1} T \\
&\geq& v_i (1-\delta)^{j-1} \delta T  + (T - B_j(v_i)) -\rho(1-\delta)^{j-1} T\\
&=& (1-\delta)^{j-1}T \left(v_i \delta - (1-\rho)\frac{\min(v_i, \rho)}{1-v_1}\right) + \gamma T\\
&=& (1-\delta)^{j-1}T (1-\rho)\left(\frac{v_i - \min(v_i, \rho)}{1-v_1}\right) + \gamma T\\
&\geq& \gamma T.
 \end{eqnarray*}
 Similarly 
 \begin{eqnarray*}
&& (T - \tau) + (v_i - \rho)(1-\delta)^{j-1} T+ v_i \tau - Tv_i \\
&\geq&T - B_j(v_i) + (v_i - \rho)(1-\delta)^{j-1} T+ v_i B_j(v_i) - Tv_i \\
&=& (B_j(v_i) - T(1-(1-\delta)^{j-1}))v_i + (T - B_{j}(v_i) - \rho(1-\delta)^{j-1}T) \\
&=& (B_j(v_i) - T(1-(1-\delta)^{j-1}))v_i - \min(v_i, \rho)\delta(1-\delta)^{j-1}T + \gamma T \\
&\geq & (B_j(v_i) - T(1-(1-\delta)^{j-1}))v_i - v_i\delta(1-\delta)^{j-1}T + \gamma T \\
&\geq& (B_j(v_i) - T(1-(1-\delta)^j))v_i + \gamma T  \\ 
&\geq& \gamma T.
 \end{eqnarray*}
 \end{proof}
 
It follows from the mean-based condition (Definition \ref{def:mb-cont}) that in the interval $[A_j, B_j(v_i)]$ the buyer with value $v_i$ will, with probability at least $(1-n\gamma)$, choose an arm currently in its 1-session (i.e. an arm with label at most $j$) and hence pay $1$ each round. Since the buyer has value $v_i$ for the item with probability $q_i$, the total contribution of the buyer with value $v_i$ to the expected revenue of the seller is given by

\begin{eqnarray*}
q_i\sum_{j=1}^{n}(1-\gamma)(B_{j}(v_i) - A_{j}(v_i)) &=& q_i\sum_{j=1}^{n}(1-n\gamma)\left(\frac{\min(u,\rho)}{1-v_1}(1-\rho)(1-\delta)^{j-1}T - \gamma T\right) \\
&=& (1-n\gamma)q_i T\left(-n \gamma + \frac{(1-\rho)\min(v_i, \rho)}{1-v_1}\sum_{j=1}^{n}(1-\delta)^{j-1}\right) \\
&=& (1-n\gamma)q_i T\left(-n \gamma + \frac{(1-\rho)\min(v_i, \rho)(1-(1-\delta)^n)}{(1-v_1)\delta}\right) \\
&=& (1-n\gamma)q_i T\left(-n \gamma + \min(v_i, \rho)(1-(1-\delta)^n)\right) \\
&=& q_i T \min(v_i, \rho)(1-(1-\delta)^n) - o(T) \\
&\geq & q_i T \left(1 - \frac{\eps}{2}\right)^2 v_i - o(T) \\
&\geq & (1-\eps)q_i v_i T - o(T).
\end{eqnarray*}

Here we have used the fact that $(1-(1-\delta)^n) = 1 - \eps/2$ (since $n = \log (\eps/2) / \log (1-\delta)$) and $\min(v_i, \rho) \geq (1-\eps/2)v_i$ (since if $\min(v_i, \rho) \neq v_i$, then $\rho = (1-\eps/2) \geq (1-\eps/2)v_i$. Summing this contribution over all $v_i \in \D$, we have that the expected revenue of the seller is at least

\begin{eqnarray*}
\sum_{i} \left((1-\eps)q_i v_i T - o(T)\right) &=& (1-\eps)\left(\sum_{i}q_iv_i\right)T - o(T) \\
&=& (1-\eps)\E_{v\sim \D}[v] T \\
&=& (1-\eps)\Val(\D)T.
\end{eqnarray*}

\end{proof}

\section{Optimal revenue against conservative buyers}\label{sect:conservative}

In Theorem \ref{thm:nc_seller}, we demonstrated a mechanism for the seller that extracts full welfare from a buyer running a mean-based learning algorithm. This mechanism, while in some sense as good as possible (it is impossible to extract more than welfare from any buyer running a \lowregret strategy), has several drawbacks. One general drawback is that it is extremely unlikely the mechanism in Section \ref{sect:nonconservative} would arise naturally as the allocation rule for any sort of auction that might arise in practice. A more specific drawback is that this mechanism assumes buyers are learning over all possible bids, instead of just bids less than their value; indeed, all arms essentially cost the maximum possible price per round, and their only difference is when they give the item away for free and when they charge for it.

In this section, we address the second drawback by studying this problem for \textit{conservative buyers}; buyers who are constrained to only submit bids less than their current value for the item. We characterize via a linear program the optimal revenue attainable for the seller when playing against conservative buyers running a mean-based learning algorithm over their set of allowable bids. We show that, while we can no longer achieve the full welfare as in Section \ref{sect:nonconservative}, we can still achieve strictly more than the Myerson revenue. Interestingly, our optimal mechanism has a natural interpretation as a repeated first-price auction with gradually decreasing reserve, thus also partially addressing the first drawback. Notably, this auction is a \textit{critical auction}. Since clever buyers act conservatively in critical auctions, this mechanism is simultaneously the optimal critical auction against clever buyers.

\subsection{Characterizing the optimal revenue}
\label{sec:chara}

\begin{figure}[h]
\begin{alignat*}{2}
  \textbf{maximize }   & \sum_{i=1}^m q_i (v_i x_i - u_i)\  \\
  \textbf{subject to ~~~~} & u_i \geq (v_i-v_j) \cdot x_j,  &\ & \forall \;i,j \in [m]: i > j\\
                       & u_i \geq 0, 1 \geq x_i \geq 0,  \ &\ & \forall \; i \in [m] \\
\end{alignat*}
\caption{The mean-based revenue LP (same as Figure \ref{fig:mblpmain}).}
\label{fig:mblp}
\end{figure}

We begin by describing the optimal strategy for the seller against mean-based conservative buyers. Fix some small constant $\eps > 0$. Recall that the buyer's value distribution $\D$ is supported on the $m$ values $0 \leq v_1 < v_2 < \dots < v_m \leq 1$, with $\Pr[v_i] = q_i$. The seller will offer $m$ options, one for each possible value. Option $i$ (corresponding to bidding $b_i = v_i$) will charge $0$ and not allocate the item for the first $(1-x_i)T$ rounds, and charge $b_i - \eps$ and allocate the item for the remaining $x_iT$ rounds. The values $x_i$ are computed by finding an optimal solution to the above LP (Figure \ref{fig:mblp}), which we call the \textit{mean-based revenue LP}. We will call the value of this LP the \textit{mean-based revenue} of $\D$, and write this as $\Crit(\D)$. Our goal in this subsection will be to show that this strategy achieves approximately $\Crit(\D)T$ total revenue against a conservative buyer running a mean-based algorithm, and that this is tight; no other strategy for a non-adaptive seller can obtain more than $\Crit(\D)T$ revenue.

To show that this is a valid strategy for the seller, we need to show that the values $x_i$ are monotone increasing. Luckily, this follows simply from the structure of the mean-based revenue LP. 

\begin{lemma}
\label{lem:mono}
Let $x_1, x_2, \dots, x_m, u_1, u_2, \dots, u_m$ be an optimal solution to the mean-based revenue LP. Then for all $i < j$, $x_i < x_j$. 
\end{lemma}
\begin{proof}
We proceed by contradiction. Suppose that the sequence of $x_i$ are not monotone; then there exists an $1 \leq i \leq m-1$ such that $x_i > x_{i+1}$. Now consider another solution of the LP, where we increase $x_{i+1}$ to $x_i$, keeping the value of all other variables the same. This new solution does not violate any constraints in the LP since for all $j > i+1$, $u_j \geq (v_j-v_i) \cdot x_i \geq (v_j - v_{i+1}) \cdot x_i$. However this change increases the value of the objective by $v_{i+1}q_{i+1}(x_{i} - x_{i+1}) > 0$, thus contradicting the fact that $x_1,\dots ,x_m, u_1,...,u_m$ was an optimal solution of the mean-based revenue LP.  
\end{proof}  

We begin by showing that this strategy achieves revenue at least $\Crit(\D)T - o(T)$ when the buyer is using a mean-based algorithm. 

\begin{theorem}[Restatement of Theorem \ref{thm:mbrevlbmain}]
\label{thm:mbrevlb}
  The above strategy for the seller gets revenue at least $(\Crit(\D) - \eps)T - o(T)$ against a conservative buyer running a mean-based algorithm. In addition, this strategy is critical. 
\end{theorem}
  
 \begin{proof}
First of all, by Lemma \ref{lem:mono}, it is easy to check the strategy is critical.
 
To prove the rest, we will show that: i) the buyer with value $v_i$ receives the item for at least $x_iT -o(T)$ turns (receiving $v_ix_iT - o(T)$ total utility from the items), and ii) this buyer's net utility is at most $(u_{i}+\eps)T + o(T)$. This implies that this buyer pays the seller at least $x_{i}v_iT - (u_{i}+\eps)T - o(T)$ over the course of the $T$ rounds; taking expectation over all $v_i$ completes the proof.

Assume the buyer is running a $\gamma$-mean-based learning algorithm. Consider the buyer when they have value $v_i$. Note that

$$\sigma_{j,t}(v_i) = (v_i - v_j + \eps)\cdot \max(0, t-(1-x_j)T).$$

We first claim that after round $(1-x_i)T + \gamma T/\eps$, the buyer will buy the item (i.e., choose an option that results in him getting the item) each round with probability at least $1-m\gamma$. To see this, first note that $\sigma_{i,t}(v_i) \geq \gamma T$ when $t \geq (1-x_i)T + \gamma T/\eps$. Then, since the cumulative utility of any arm is $0$ until it starts offering the item, it follows from the mean-based condition that the buyer will pick a specific arm that is not offering the item with probability at most $\gamma$, and therefore choose some good arm with probability at least $1-m\gamma$. It follows that, in expectation, the buyer with value $v_i$ receives the item for at least $(1-m\gamma)(x_i T - \gamma T/\eps) = x_i T - o(T)$ turns.

We now proceed to upper bound the overall expected utility of the buyer. For each index $j \leq i$, let $S_j$ be the set of $t$ where $\sigma_{j,t}(v_i) > \sigma_{j',t}(v_i)$ for all other $j'$. Note that since each $\sigma_{j,t}(v_i)$ is a linear function in $t$ (when positive), each $S_j$ is either the empty set or an interval $(y_jT, z_jT)$. Since all the $v_i$ are distinct, note that these intervals partition the interval $((1-x_i)T, T)$ (with the exception of up to $m$ endpoints of these intervals); in particular, $\sum_{j\geq i} (z_j - y_j) = x_i$. 

Let $\eps' = \min_{j}(v_{j+1}-v_{j})$. Note that, if $t \in (y_jT + \gamma T/\eps', z_jT - \gamma T/\eps')$, then for all $j' \neq j$, $\sigma_{j, t}(v_i) > \sigma_{j', t}(v_i) + \gamma T$. This follows since $\sigma_{j, t}(v_i) - \sigma_{j', t}(v_i)$ is linear in $t$ with slope $v_{j} - v_{j'}$, and $|v_j - v_{j'}| > \eps'$. It follows that if $t$ is in this interval, then the buyer will choose option $j$ with probability at least $1-m\gamma$ (by a similar argument as before).

Define $j(t) = \arg\max_{j} \sigma_{j,t}(v_i)$ to be the index of the arm with the current largest cumulative reward, and let $\sigma_{max, t}(v_i) = \sum_{s=1}^{t}r_{j(s), s}(v_i)$ be the cumulative utility of always playing the arm with the current highest cumulative reward for the first $t$ rounds. The following lemma shows that $\sigma_{max, T}(v_i)$ is close to $\max_{j}\sigma_{j, T}(v_i)$. (In other words, playing the best arm every round and playing the best-at-the-end arm every round have similar payoffs if the historically best arm does not change often).
\begin{lemma}
$|\sigma_{max, T}(v_i) - \max_{j}\sigma_{j, T}(v_i)| \leq m$.
\end{lemma}
\begin{proof}
Let $W = |\{t | j(t) \neq j(t+1)\}|$ equal the number of times the best arm switches values; note that since each $\sigma_{j, t}(v_i)$ is linear, $W$ is at most $m$. Let $t_1 < t_2 < \dots < t_{W}$ be the values of $t$ such that $j(t) \neq j(t+1)$. Additionally define $t_0 = 1$ and $t_{W+1} = T$. Then, dividing the cumulative reward $\sigma_{max, t}$ into intervals by these $t_i$, we get that

\begin{eqnarray*}
\sigma_{max, t}(v_i) &=& \sum_{s=1}^{t}r_{j(s), s}(v_i)\\
&=& \sum_{i=1}^{W+1}(\sigma_{j(t_i), t_i}(v_i) - \sigma_{j(t_i), t_{i-1}}(v_i)) \\
&=& \sigma_{j(T), T}(v_i) + \sum_{i=1}^{W+1}(\sigma_{j(t_{i-1}), t_{i-1}}(v_i)- \sigma_{j(t_i), t_{i-1}}(v_i)) \\
&=& \max_{j}\sigma_{j, t}(v_i) + \sum_{i=1}^{W+1}(\sigma_{j(t_{i-1}), t_{i-1}}(v_i)- \sigma_{j(t_i), t_{i-1}}(v_i))
\end{eqnarray*}

It therefore suffices to show that $|\sigma_{j(t_{i-1}), t_{i-1}}(v_i)- \sigma_{j(t_i), t_{i-1}}(v_i)| \leq 1$ for all $i$. To see this, note that (by the definition of $j(t)$), $\sigma_{j(t_{i-1}), t_{i-1}}(v_i)- \sigma_{j(t_i), t_{i-1}}(v_i) > 0$, and that $\sigma_{j(t_{i-1}), t_{i-1}+1}(v_i)- \sigma_{j(t_i), t_{i-1}+1}(v_i) < 0$. However, 

$$(\sigma_{j(t_{i-1}), t_{i-1}+1}(v_i)- \sigma_{j(t_i), t_{i-1}+1}(v_i)) = (\sigma_{j(t_{i-1}), t_{i-1}}(v_i)- \sigma_{j(t_i), t_{i-1}}(v_i)) + (r_{j(t_{i-1}), t_{i-1}+1}(v_i) - r_{j(t_{i}), t_{i-1}+1}(v_i))$$

Since $0 \leq r_{j, t}(u) \leq 1$, it follows that $|\sigma_{j(t_{i-1}), t_{i-1}}(v_i)- \sigma_{j(t_i), t_{i-1}}(v_i)| \leq 1$. This completes the proof.
\end{proof}

Let $\sigma_{T}(v_i) = \sum_{t=1}^{T} \E[r_{I_{t}, t}(v_i)]$ denote the expected cumulative utility of this buyer at time $T$. We claim that $\sigma_{T} \leq \max_{j}\sigma_{j, T}(v_i) + o(T)$. To see this, recall that, for $t \in (y_{j}T + \gamma T/\eps', z_{j}T - \gamma T/\eps')$, $\Pr[I_{t} \neq j] \leq m\gamma$, and therefore $\E[r_{I_t, t}] \leq r_{j,t} + m\gamma$. Furthermore, note that for $t \in S_j$, $j(t) = j$, so $r_{j,t} = r_{j(t), t}$ and $\E[r_{I_t, t}] \leq r_{j(t), t} + m\gamma$. It follows that

\begin{eqnarray*}
\sigma_{T}(v_i) &=& \sum_{t=1}^{T} \E[r_{I_{t}, t}(v_i)] \\
&\leq & \sum_{t=(1-x_i)T}^{T} \E[r_{I_{t}, t}(v_i)] \\
&=& \sum_{j=1}^{i} \sum_{t=y_{j}T}^{z_{j}T} \E[r_{I_{t}, t}(v_i)] \\
&\leq & \sum_{j=1}^{i}\left(\frac{2\gamma T}{\eps'} + \sum_{t=y_{j}T+\gamma T/\eps'}^{z_{j}T - \gamma T/\eps'} \E[r_{I_{t}, t}(v_i)]\right) \\
&\leq & \sum_{j=1}^{i}\left(\frac{2\gamma T}{\eps'} + \sum_{t=y_{j}T+\gamma T/\eps'}^{z_{j}T - \gamma T/\eps'} ( r_{j(t), t}(v_i) + m\gamma)\right) \\
&\leq& \frac{2m\gamma T}{\eps'} + m\gamma T + \sum_{t=1}^{T} r_{j(t), t}(v_i) \\
& = & \frac{2m\gamma T}{\eps'} + m\gamma T + \sigma_{max, T}(v_i) \\
& \leq & \frac{2m\gamma T}{\eps'} + m\gamma T + m + \max_{j}\sigma_{j, T}(v_i) \\ 
& = & \max_{j}\sigma_{j, T}(v_i) + o(T).
\end{eqnarray*}

Finally, note that

\begin{eqnarray*}
\max_{j}\sigma_{j, T}(v_i) &=& \max_{j < i} (v_i - v_j + \eps)x_{j}T \\
&\leq & (\max_{j < i} (v_i - v_j)x_{j} + \eps)T \\
&= & (u_i + \eps)T
\end{eqnarray*}

It follows that $\sigma_{T}(v_i) \leq (u_i + \eps)T + o(T)$, as desired.

\end{proof}
    
We now proceed to show that this bound is in fact optimal; no strategy for the seller (even an adaptive one) can achieve better revenue against a \lowregret, conservative buyer.

\begin{theorem}[Restatement of Theorem \ref{thm:mbrevubmain}]
\label{thm:mbrevub}
Any strategy for the seller achieves revenue at most $\Crit(\D)T + o(T)$ against a conservative buyer running a \lowregret algorithm.
\end{theorem}  
\begin{proof}
Assume the buyer is running a $\delta$-\lowregret algorithm, for some $\delta = o(T)$. Consider an arbitrary strategy for the seller with $K$ arms, where arm $j$ is labelled with maximum bid $b_j$. We begin by claiming that the following LP (Figure \ref{fig:lp1})  provides an upper bound on the revenue obtainable by this strategy against our \lowregret buyer.

\begin{figure}[h]
  \begin{alignat*}{2}
    \textbf{maximize }   & \sum_{i=1}^m q_i(v_i x_i - u_i)\  \\
    \textbf{subject to ~~~~} & u_i \geq v_i y_j -\overline{p}_j - \delta/T, &\ & i\in [m],j\in[K]: v_i \geq b_j\\
                       & \overline{p}_j \leq b_j y_j,  \ &\ & j \in [K]\\
                       & x_i = y_j, \ &\ & i \in [m], j = \arg\max_{j \in [K]:b_j \leq v_i} b_j \\
                       & \overline{p}_j \geq 0, 1 \geq y_j \geq 0,\ &\ & j \in [K] \\
  \end{alignat*}
  \caption{$LP'$, with variables $x_i$, $u_i$, $y_j$, and $\overline{p}_j$}
\label{fig:lp1}
\end{figure}
  
  \begin{lemma}\label{lem:revub}
  Let $V'$ be the optimal value of $LP'$ (see Figure \ref{fig:lp1}). Then the expected revenue of the seller is at most $V'T$.
  \end{lemma}
\begin{proof}
Given our strategy for the seller, we will assign values to variables in the following way. Fix a strategy for the buyer, and let $y_j = \frac{1}{T}\E\left[\sum_{t}a_{j,t}\right]$ be the expected average probability that arm $j$ gives the item and let $\overline{p}_j = \frac{1}{T}\E\left[\sum_{t}p_{j,t}\right]$ be the expected average price charged by arm $j$. We will define $x_i$ through the third constraint, and set $u_i = \max_{j}(v_iy_j - \overline{p}_j - \delta/T)$. We will show that this assignment of variables satisfies all the constraints, and that the objective function evaluated on this assignment of variables is at least the seller's revenue using this strategy.

The first and third constraints are satisfied via our choices of $x_i$ and $u_i$. The constraint $\overline{p}_{j} \leq b_jy_j$ is satisfied since $p_{j, t} \leq b_{j}a_{j,t}$ for all $t$. Finally, $0 \leq y_j \leq 1$ is satisfied since $y_j$ is an average probability.

We now must show that the seller's revenue is at most $q_i(v_ix_i - u_i)$. We begin by claiming that $x_i$ is an upper bound for the expected fraction of the time that the buyer receives the item when he has value $v_i$. To see this, note first that the buyer is conservative, and therefore will not bid on any arm with bid value larger than $v_i$. Choose $j$ so that $b_j$ is maximized over all $b_j \leq v_i$; note that since the seller's strategy is monotone, $a_{j,t} > a_{j',t}$ for any $j' < j$, so the buyer will receive the item at most $\E\left[\frac{1}{T}\sum_{t} a_{j,t}\right] = y_j$ of the time in expectation. But by our third constraint, $x_i = y_j$, so $x_i$ is an upper bound on the average probability that the buyer with value $v_i$ gets the item, and therefore $ \sum_{i=1}^m q_iv_ix_i$ is an upper bound on the average welfare of the buyer. 

We next claim that $\sum_{i}q_iu_i$ is a lower bound for the average utility of the buyer. To see this, note that since the buyer is using a $\delta$-\lowregret algorithm, when the value is $v_i$, the buyer should not regret always playing some arm $j$ with $w_j \leq v_i$. Therefore the average surplus of value $v_i$ should satisfy the constraint on $u_i$, and so $ \sum_{i=1}^m q_i \cdot  u_i$ is a lower bound on the average surplus of the buyer.

Finally, note that the seller's revenue is just the buyer's welfare minus the buyer's surplus. Combining the upper bound on the buyer's welfare and the lower bound on the buyer's surplus, we get our desired upper bound on the seller's revenue. 
\end{proof}
  
  We will now show how to transform a solution of this LP into a solution to the mean-based revenue LP while ensuring that its value does not decrease by more than $\delta/T$. To begin, it is easy to see that there exists an optimal solution of $LP'$ that satisfies $\overline{p}_j = y_j \cdot w_j$ for all $j \in [K]$. We can thus increase each $u_i$ by $\delta/T$, since this will decrease the value of the LP by at most $\delta/T$ as $\sum_{i=1}^m q_i = 1$. This solution now satisfies $u_i \geq (v_i - b_j)y_j$ for all $i\in [m], j\in [K]:v_i \geq b_j$. Finally, for each $i,j \in [m]: i > j$, note that for $\ell = \arg \max_{\ell \in [K]:b_{\ell} \leq v_j} b_\ell$, we have that $b_{\ell} \leq v_j$. It follows that $u_i \geq (v_i - v_j)y_{\ell} = (v_i - v_j)x_j$, and therefore that this solution is a valid solution of the mean-based revenue LP. 
  
From the above argument, we can conclude that $V_1 \leq R_{mb}(\D) + \delta/T$. It follows from Lemma \ref{lem:revub} that the total revenue is upper bounded by $T(\Crit(\D) + \delta/T) = R_{mb}(\D)T + o(T)$, as desired.
\end{proof}

Note that the proof of Lemma \ref{lem:revub} relies on the fact that our allocation rule is monotone. We can show that this constraint is necessary; with non-monotone strategies, the seller can extract up to the full welfare of a conservative buyer playing a mean-based strategy. The proof of this fact can be found in Appendix \ref{sect:non-monotone}.
 
  \subsection{Bounding $\Crit(\D)$}
  \label{sec:boundmb}
In this section, we compare the mean-based revenue $\Crit(\D)$ to our two benchmarks: the Myerson revenue for the item, $\Mye(\D)$, and the buyer's expected value for the item, $\Val(\D)$. It is not too hard to see that $\Crit(\D) \leq \Val(\D)$ (the value of the mean-based revenue LP is clearly at most $\sum_{i}q_iv_i = \Val(\D)$) and that $\Crit(\D) \geq \Mye(\D)$ (the seller can achieve $\Mye(\D)$ by just always selling the item the Myerson price). We show here that $\Crit(\D)$ is not a constant factor approximation to either $\Mye(\D)$ or $\Val(\D)$, and thus lies strictly between our two benchmarks in general.

We will begin by showing that $\Crit(\D)$ is monotone with respect to stochastic dominance. We will break from notation somewhat by considering distributions $\D$ supported on $[1, H]$ rather than $[0, 1]$; since $\Mye(\D)$, $\Crit(\D)$, and $\Val(\D)$ are all linear in the values $v_i$, dividing all values through by $H$ results restores the condition that $\D$ is supported on $[0,1]$ while preserving the multiplicative gaps between these quantities.

\begin{definition}
A distribution $\mathcal{D}$ stochastically dominates distribution $\mathcal{D'}$ if for all $t$, $\Pr_{u\sim \mathcal{D}}[u \geq t] \geq \Pr_{u\sim \mathcal{D}'}[u \geq t]$. 
\end{definition}

\begin{lemma}
If distribution $\D$ stochastically dominates distribution $\D'$, then $\Crit(\D) \geq \Crit(\D')$.
\end{lemma}
\begin{proof}

Note that we can write $\Crit(\D)$ in the form

$$\Crit(\D) = \max_{x}\mathbb{E}_{v_i \sim \D}\left[v_{i}x_{i} - \max_{j}(v_{i}-v_{j})x_{j} \right]$$

To show $\Crit(\D) \geq \Crit(\D')$, it suffices to show that for all increasing $x$ (i.e. $x_{i} \geq x_{j}$ for $i \geq j$), that

$$\mathbb{E}_{v_i \sim \D}\left[v_{i}x_{i} - \max_{j}(v_{i}-v_{j})x_{j} \right] \geq \mathbb{E}_{v_i \sim \D'}\left[v_{i}x_{i} - \max_{j}(v_{i}-v_{j})x_{j} \right]$$

Note that if $\mathcal{D}$ stochastically dominates distribution $\mathcal{D'}$, then for any increasing function $f$, $\mathbb{E}_{u \sim \mathcal{D}}[f(u)] \geq \mathbb{E}_{u \sim \mathcal{D'}}[f(u)]$. It suffices to show that $f(v_{i}) = v_{i}x_{i} - \max_{j}(v_{i}-v_{j})x_{j}$ is increasing in $i$ (and hence in $v_{i}$). In particular, we wish to show that, for $i' > i$,

$$v_{i'}x_{i'} - \max_{j}(v_{i'} - v_{j})x_{j} \geq v_{i}x_{i} - \max_{j}(v_{i'} - v_{j})x_{j}$$

\noindent
or equivalently,

$$\min_{j} \left( v_{i'}x_{i'} - (v_{i'} - v_{j})x_{j}\right) \geq \min_{j} \left(v_{i}x_{i} - (v_{i} - v_{j})x_{j}\right).$$

To show this, it suffices to show that for each $j$, 

$$v_{i'}x_{i'} - (v_{i'} - v_{j})x_{j} \geq v_{i}x_{i} - (v_{i} - v_{j})x_{j} $$

\noindent
or equivalently,

$$v_{i'}x_{i'} - v_{i}x_i \geq (v_{i'} - v_{i})x_{j}.$$

This follows since

\begin{eqnarray*}
v_{i'}x_{i'} - v_{i}x_i &\geq & v_{i'}x_i - v_{i}x_i \\
&= & (v_{i'} - v_i)x_i \\
&\geq & (v_{i'} - v_{i})x_j.
\end{eqnarray*}

Here we have used the fact that $x_{i'} \geq x_{i} \geq x_{j}$. This concludes the proof.

\end{proof}

For ease of analysis, we will also switch to considering continuous distributions $\D$. The definitions of $\Mye(\D)$ and $\Val(\D)$ still hold for continuous $\D$. Since the mean-based revenue LP implies that, in the optimal solution, $u_{i} = \max_{j}(v_i - v_j)x_j$, we can write $\Crit(\D)$ for a continuous $\D$ supported on $[1, H]$ with pdf $q(v)$ as

$$\Crit(\D) = \max_{x(v)} \int_{1}^{H}q(v)(vx(v) - \max_{w<u}(v-w)x(w))dv.$$

By discretizing appropriately, all gaps we prove for continuous $\D$ extend to discrete values of $\D$.

\begin{definition}
The \emph{equal revenue curve} is the (continuous) distribution $\D_{ERC}$ supported on $[1, \infty)$ with CDF $F(v) = 1 - \frac{1}{v}$. The \emph{equal revenue curve truncated at $H$} is the distribution distribution $\D_{ERC}(H)$ supported on $[1, H]$ with CDF $F(v) = 1 - \frac{1}{v}$ for $v \leq H$ and $F(v) = 0$ for $v > H$.
\end{definition}

Note that $\Mye(\D_{ERC}) = 1$ (since $v(1-F(v)) = 1$ for all $v \geq 1$). Likewise, $\Mye(\D_{ERC}(H)) = 1$. 

\begin{lemma}
Let $\D_{ERC}(H)$ be the equal revenue curve truncated at $H$. Let $\D$ be any distribution supported on $[1, H]$ with $\Mye(D) = 1$. Then $\D_{ERC}(H)$ stochastically dominates $\D$.
\end{lemma}

\begin{corollary}
The distribution $\D$ supported on $[1, H]$ that maximizes $\Crit(\D)$ subject to $\Mye(\D) = 1$ is the truncated equal revenue curve $\D_{ERC}(H)$. 
\end{corollary}

\begin{theorem}\label{thm:gaplb}
$\Crit(\D_{ERC}(H)) \geq \Omega(\log\log H)$.
\end{theorem}
\begin{proof}
Note that for $\D_{ERC}(H)$, the pdf $q(v)$ is given by $q(v) = \frac{1}{v^2}$, so

\begin{eqnarray*}
\Crit(\D_{ERC}(H)) &\geq & \max_{x(v)}\int_{1}^{H}q(v)(vx(v) - \max_{w < v}(v-w)x(w))dv \\ 
&=& \max_{x(v)}\int_{1}^{H}\frac{1}{v}\left(x(v) - \max_{w<v}\left(1 - \frac{w}{v}\right)x(w)\right)dv.
\end{eqnarray*}

Here the maximum of $x(v)$ is taken over all increasing functions from $[1, H]$ to $[0, 1]$. Consider the function $x(v) = \frac{\log v}{\log H}$. In this case, $(v-w)x(w)$ is maximized when:

\begin{eqnarray*}
\frac{d}{dv}\left((v-w)x(w)\right) &=& 0 \\
(v-w)x'(w) - x(w) &=& 0 \\
(v-w)\frac{1}{w \log H} - \frac{\log w}{\log H} &=& 0 \\
w + w\log w &=& v.
\end{eqnarray*}

If we choose $w$ so that the above inequality holds, then note that $dv = (2 + \log w)dw$. It follows that

\begin{eqnarray*}
&&\Crit(\D_{ERC}(H)) \\
&\geq & \frac{1}{\log H}\int_{1}^{H}\frac{1}{w+w\log w}\left(\log(w+w\log w) - \left(1 - \frac{w}{w+w\log w}\right)\log w\right)(2+\log w)dw \\
&=& \frac{1}{\log H}\int_{1}^{H}\frac{(2+\log w)}{w+w\log w}\left(\log(w+w\log w) - \log w + \frac{\log w}{1+\log w}\right)dw \\
&\geq & \frac{1}{\log H}\int_{1}^{H}\frac{(2+\log w)}{w+w\log w}\log(1+\log w)dw \\
&\geq & \frac{1}{\log H}\int_{1}^{H}\frac{\log(1+\log w)}{w}dw \\
&=& \frac{\log(H)\log(1+\log H) - \log(1 + \log H) - \log H}{\log H}\\
&=& \Omega(\log\log H)
\end{eqnarray*}

\end{proof}

\begin{theorem}\label{thm:gapub}
$\Crit(\D_{ERC}(H)) \leq O(\log\log H)$.
\end{theorem}
\begin{proof}
Note that, up to a point mass at $H$ which contributes at most $H(1/H) = 1$ to the mean-based revenue, $\Crit(\D_{ERC}(H))$ is given by

$$\max_{x(v)}\int_{1}^{H}\frac{1}{v}\left(x(v) - \max_{w<v}\left(1 - \frac{w}{v}\right)x(w)\right)dv.$$

Let $f(v): [1, \infty) \rightarrow [1, \infty)$ be a function that satisfies $f(v) < v$ for all $v \in [1, \infty)$. By choosing $w = f(v)$, we have that

\begin{eqnarray*}
\Crit(\D_{ERC}(H)) &\leq & \max_{x(v)} \left(\int_{1}^{H}\frac{1}{v}\left(x(v) - \left(1 - \frac{f(v)}{v}\right)x(f(v))\right)dv \right)\\
&=& \max_{x(v)} \left(\int_{1}^{H}\frac{x(v)}{v}dv - \int_{1}^{H}\frac{1}{v}\left(1 - \frac{f(v)}{v}\right)x(f(v))dv \right) \\
&=& \max_{x(v)} \left(\int_{f(H)}^{H}\frac{x(v)}{v}dv + \int_{1}^{f(H)}\frac{x(v)}{v}dv - \int_{1}^{H}\left(\frac{1}{v} - \frac{f(v)}{v^2}\right)x(f(v))dv \right) \\
&=& \max_{x(v)} \left(\int_{f(H)}^{H}\frac{x(v)}{v}dv + \int_{1}^{H}\frac{x(f(v))f'(v)}{f(v)}dv - \int_{1}^{H}\left(\frac{1}{v} - \frac{f(v)}{v^2}\right)x(f(v))dv \right) \\
&=& \max_{x(v)} \left(\int_{f(H)}^{H}\frac{x(v)}{v}dv + \int_{1}^{H}\left(\frac{f'(v)}{f(v)} + \frac{f(v)}{v^2} - \frac{1}{v}\right)x(f(v))dv \right). \\
\end{eqnarray*}

Choose $f(v) = \frac{v}{1+\log{v}}$. Note that, for this choice of $f$,

$$f'(v) = \frac{\log v}{(1+\log v)^2},$$

and so

\begin{eqnarray*}
\frac{f'(v)}{f(v)} + \frac{f(v)}{v^2} - \frac{1}{v} &=& \frac{\log v}{v(1+\log v)} + \frac{1}{v(1+\log v)} - \frac{1}{v} \\
&=& 0.
\end{eqnarray*}

It follows that (since $x(v) \in [0, 1]$ for all $v$)

\begin{eqnarray*}
\Crit(\D_{ERC}(H)) &\leq & \max_{x(v)}\int_{f(H)}^{H} \frac{x(v)}{v}dv \\
&\leq & \int_{H/(\log H + 1)}^{H} \frac{dv}{v} \\
&=& \log(\log H + 1) \\
&=& O(\log\log H).
\end{eqnarray*}

\end{proof}

\begin{corollary}[Restatement of Theorem \ref{thm:mbrev}]
The gap $\Crit(\D)/\Mye(\D)$ can grow arbitrarily large. For distributions $\D$ supported on $[1, H]$, this gap can be as large as $\Omega(\log \log H)$ (and this is tight). Similarly, the gap $\Val(\D)/\Crit(\D)$ can grow arbitrarily large. For distributions $\D$ supported on $[1,H]$, this gap can be as large as $\Omega(\log H / \log \log H)$.
\end{corollary}

\section{Mean-based learning algorithms}\label{sect:mbalgs}

In this appendix we will show that Multiplicative Weights and EXP3 - the most common adversarial \lowregret algorithms for the experts and bandits case respectively - are mean-based, as per Definition \ref{def:mb}. We expect that many variants of these algorithms along with other \lowregret learning algorithms are also mean-based, and can be shown to be mean-based via similar methods of proof.

We begin by showing that Multiplicative Weights (Algorithm \ref{alg:mw}) is mean-based. Multiplicative Weights, also known as Hedge (see survey \cite{AroraHK12} for more details) is a simple \lowregret learning algorithm for the full-information setting. It proceeds by maintaining a weight $w_i$ for each option. Every round, Multiplicative Weights chooses an option with probability proportional to $w_i$, and then updates each weight $w_i$ by multiplying it by $e^{\eps r_i}$, where $\eps$ is a parameter of the algorithm and $r_i$ is the reward from option $i$ this round.

\begin{algorithm}[ht] 
	\caption{Multiplicative Weights algorithm.} \label{alg:mw}
    \begin{algorithmic}[1]
    	\STATE Choose $\eps = \sqrt{\frac{\log K}{ T}}$. Initialize $K$ weights, letting $w_{i, t}$ be the value of the $i$th weight at round $t$. Initially, set all $w_{i,0} = 1$.
        \FOR {$t=1$ to $T$}
			\STATE Choose option $i$ with probability $p_{i,t} = w_{i,t-1} / \sum_{j}w_{j,t-1}$.
			\FOR {$j=1$ to $K$}
				\STATE Set $w_{j,t} = w_{j,t-1} \cdot e^{\eps r_{j,t}}$.
			\ENDFOR
        \ENDFOR
      \end{algorithmic}
\end{algorithm}

\begin{theorem}\label{thm:mbmw}
The Multiplicative Weights algorithm (Algorithm \ref{alg:mw}) is mean-based.
\end{theorem}
\begin{proof}
Define $\gamma =2(T\eps)^{-1} \log(T\eps)$. We will show that Multiplicative Weights is $\gamma$-mean-based. Note that since $\eps = \sqrt{\frac{\log K}{ T}}$, $\gamma = o(1)$ and therefore Multiplicative Weights is mean-based.

Note that $w_{i,t} = e^{\eps \sigma_{i,t}}$. Therefore, if $\sigma_{i,t} - \sigma_{j,t} < -\gamma T$, we have $\sigma_{i,t-1} - \sigma_{j,t-1} < -\gamma T + 1 < -\gamma T/2$, it follows that
\begin{eqnarray*}
p_{i,t} &=& \frac{w_{i,t-1}}{\sum_{j}w_{j,t-1}} \\
&\leq & \frac{w_{i, t-1}}{w_{j,t-1}} \\
&=& e^{\eps(\sigma_{i,t-1} - \sigma_{j,t-1})} \\
&< & e^{-\eps \gamma T/2} \\
&=& e^{-\log(T\eps)} = 1/(T\eps) \leq \gamma.
\end{eqnarray*}

It follows that Multiplicative Weights is $\gamma$-mean-based.
\end{proof}

\begin{algorithm}[ht] 
	\caption{Follow-the-Perturbed-Leader algorithm.} \label{alg:ftpl}
    \begin{algorithmic}[1]
    	\STATE Choose $\eps = \sqrt{\frac{\log K}{ T}}$. 
        \FOR {$t=1$ to $T$}
        			\STATE For each arm, sample $per_i\geq 0$ independently from exp. distribution $d\mu(x) = \varepsilon e^{-\varepsilon x}$. 
			\STATE Choose option $i$ with largest $\sigma_{i,t-1} + per_i$.
        \ENDFOR
      \end{algorithmic}
\end{algorithm}

We now show the Follow-the-Perturbed-Leader algorithm (Algorithm \ref{alg:ftpl}) is mean-based.
\begin{theorem}\label{thm:mbftpl}
The Follow-the-Perturbed-Leader algorithm (Algorithm \ref{alg:ftpl}) is mean-based.
\end{theorem}
\begin{proof}
Let $\gamma = \sqrt{\frac{1}{T}} \cdot \log(T)$. When $\sigma_{i,t} < \sigma_{j,t} - \gamma T$, the probability option $i$ is chosen at round $i$ is at most
\[
Pr[ per_i > \sigma_{i,t-1} - \sigma_{j,t-1}]  = e^{-\varepsilon ( \sigma_{i,t-1} - \sigma_{j,t-1})} \leq e^{-\varepsilon \gamma T /2} < \sqrt{\frac{1}{T}} < \gamma.
\]
Therefore the Follow-the-Perturbed-Leader algorithm (Algorithm \ref{alg:ftpl}) is $\gamma$-mean-based.
\end{proof}

We will now show that EXP3 (Algorithm \ref{alg:exp3}) is mean-based. EXP3 can be thought of as an extension of Multiplicative Weights to the incomplete information (bandits) setting. Since we no longer observe every option's reward each round, we cannot perform the same weight update rule as in Multiplicative Weights. Instead, if we choose option $i$, we update weight $w_i$ by multiplying it with $e^{\eps r_i/p_i}$, where $p_i$ is the probability of picking this option this round (i.e. $w_i/\sum w_j$), and leave all other weights unmodified. Since $\E[\frac{r_{i,t}}{p_{i,t}}\mathbbm{1}_{I_{t}=i}] = r_{i,t}$, this accomplishes in expectation (in some sense) the same update rule as Multiplicative Weights. It is known that (for fixed $K$) if $\eps = T^{-\alpha}$ for some $\alpha \in (0, 1)$, then EXP3 is \lowregret (\cite{AuerCNS03}). This regret is minimized when $\alpha = 1/2$, but for convenience of analysis we will show that EXP3 is mean-based when $\alpha = 1/4$. EXP3 is still \lowregret when $\alpha = 1/4$.

\begin{algorithm}[ht] 
	\caption{EXP3 algorithm.} \label{alg:exp3}
    \begin{algorithmic}[1]
    	\STATE Choose a parameter $\eps \in (0,1)$. Initialize $K$ weights, letting $w_{i, t}$ be the value of the $i$th weight at round $t$. Initially, set all $w_{i,0} = 1$.
        \FOR {$t=1$ to $T$}
			\STATE Choose option $i$ with probability $p_{i,t} = (1-K\eps)\frac{w_{i,t-1}}{ \sum_{j}w_{j,t-1}} + \eps$.
			\STATE Set $w_{i,t} = w_{i,t-1} \cdot e^{\eps r_{i,t}/p_{i,t}}$.
        \ENDFOR
      \end{algorithmic}
\end{algorithm}

\begin{theorem}\label{thm:mbexp3}
The EXP3 algorithm (Algorithm \ref{alg:exp3}) is mean-based.
\end{theorem}
\begin{proof}
We will set $\eps = T^{-1/4}$ and $\gamma = 2(2\sqrt{2} + 1)T^{-1/4}\log T$. We will show that EXP3 is $\gamma$-mean-based.

Define $\hat{\sigma}_{i,t} = \sum_{s=1}^{t} \frac{r_{i,s}}{p_{i,s}}\cdot \mathbbm{1}_{I_{s} = i}$. Note that $\hat{\sigma}_{i,t} - \sigma_{i,t}$ is a martingale in $t$; indeed, conditioned on the actions from time $1$ up to time $t-1$, $\E\left[\frac{r_{i,s}}{p_{i,s}} \cdot \mathbbm{1}_{I_s = i}\right] = r_{i,s}$. In addition, note that $\left| \frac{\eps r_{i,s}}{p_{i,s}} \cdot \mathbbm{1}_{I_s = i} - \eps r_{i,s} \right| \leq \frac{1}{p_{i,s}} \leq 1/\eps$, since $p_{i,s} \geq \eps$ by definition. It follows from Azuma's inequality that, for any $1 \leq i \leq K$, $1 \leq t \leq T$, and $M > 0$, 

$$\Pr\left[ |\hat{\sigma}_{i, t} - \sigma_{i,t}| \geq M \right] \leq 2\exp\left(-\frac{M^2\eps^2}{2T}\right).$$

We will choose $M$ so that $M\eps = \sqrt{2T\log T}$; for this $M$, it follows that

$$\Pr\left[ |\hat{\sigma}_{i, t} - \sigma_{i,t}| \geq M \right] \leq \frac{2}{T}.$$

Now, note that $w_{i,t} = e^{\eps\hat{\sigma}_{i,t}}$. If $\sigma_{i,t} - \sigma_{j,t} < -\gamma T$, we have $\sigma_{i,t-1} - \sigma_{j,t-1} < -\gamma T + 1 < -\gamma T/2$,  it then follows that 
\begin{eqnarray*}
p_{i,t} &=& (1-K\eps)\frac{w_{i,t-1}}{\sum_{j}w_{j,t-1}} + \eps \\
&\leq & \min\left(\frac{w_{i, t-1}}{w_{j,t-1}}, 1\right) + \eps \\
&=& \min(e^{\eps(\hat{\sigma}_{i,t-1} - \hat{\sigma}_{j,t-1})}, 1) + \eps \\
&\leq & e^{\eps(\sigma_{i,t-1} - \sigma_{j,t-1}) + 2M\eps} + \frac{2}{T} + \eps \\
&< & e^{-\eps\gamma T/2 + 2\sqrt{2T\log T}} + \frac{2}{T} + \eps \\
&\leq & e^{-\sqrt{T}\log T} + \frac{2}{T} + T^{-1/4} \\
&\leq & \gamma.
\end{eqnarray*}

\end{proof}

Finally, we prove Theorem \ref{thm:mean-based alg}, showing that the contextualization of a mean-based algorithm is still mean-based. In particular, the contextualizations of the above three algorithms (Multiplicative Weights, Follow the Perturbed Leader, and EXP3) are all mean-based algorithms for the contextual bandits problem.

\begin{theorem}[Restatement of Theorem \ref{thm:mean-based alg}]
If an algorithm for the experts problem or multi-armed bandits problem is mean-based, then its contextualization is also a mean-based algorithm for the contextual bandits problem.
\end{theorem}
\begin{proof}
Assume $M$ is a $\gamma$-mean-based algorithm. We will show $M'$ is $\frac{1}{\min_c \Pr[c]} \left(\gamma+\frac{2\sqrt{\log(mKT)}}{T^{1/2}}\right)$-mean-based.

First define $\hat{\sigma}_{i,t}(c) = \sum_{s: s \leq t,\, c_{s} = c}r_{i,s}(c)$ to be the total reward given by arm $i$ on rounds where the context is $c$. Since $M$ is $\gamma$-mean-based, whenever $\hat{\sigma}_{i, t}(c) < \hat{\sigma}_{j,t}(c) - \gamma T$, then the probability $p_{i,t}(c)$ that the algorithm pulls arm $i$ on round $t$ if it has context $c$ satisfies $p_{i,t}(c) < \gamma$.

We will proceed to show that  $\hat{\sigma}_{i, t}(c) < \hat{\sigma}_{j,t}(c) - \gamma T$ with sufficiently large probability. It is easy to check that $\E[\hat{\sigma}_{i,t}(c)] = \sigma_{i,t}(c) \cdot \Pr[c]$. By the Chernoff bound, we have that
\[
\Pr\left[ \left|\hat{\sigma}_{i,t}(c)- \sigma_{i,t}(c) \cdot \Pr[c]\right| \geq \sqrt{T\log(mKT)} \right] \leq 2\exp(-2T\log(mKT)/t) \leq \frac{2}{T^2 m^2 K^2}. 
\]
By the union bound, with probability at least $\frac{2}{Tm^2K^2}$, we have $\left|\hat{\sigma}_{i,t}(c)- \sigma_{i,t}(c) \cdot \Pr[c]\right| \geq \sqrt{T\log(mKT)} $ for all $i$,$t$, and $c$. In this case we have that $\sigma_{i, t}(c) < \sigma_{j,t}(c) - \frac{1}{\Pr[c]} (\gamma T + 2\sqrt{T\log(mKT)})$ implies that $\hat{\sigma}_{i, t}(c) < \hat{\sigma}_{j,t}(c) - \gamma T$.

Therefore, if $\sigma_{i, t}(c) < \sigma_{j,t}(c) - \frac{1}{\Pr[c]} (\gamma T + 2\sqrt{T\log(mKT)})$ and the context of round $t$ is $c$, then $p_{i,t}(c) < \gamma + \frac{2}{Tm^2K^2} \leq (\frac{1}{\min_c \Pr[c]} (\gamma+\frac{1}{T^{1/2}}))$. 
\end{proof} 

\section{Full revenue can be achieved with non-monotone strategies}\label{sect:non-monotone}

\begin{theorem}
For any constant $\varepsilon >0 $, there exists a non-monotone strategy for the seller using $O(m^2/\varepsilon)$ arms that gets revenue at least $(1-\eps)\Val(\D)T - o(T)$ against a conservative buyer running a $\gamma$-mean-based algorithm with $\gamma = o\left(\min(\varepsilon/m^2, \varepsilon/(m \cdot v_1)\right)$.
\end{theorem}

\begin{proof}
Let $M = m/\varepsilon$ and $\delta = 2\gamma/v_1$. Consider the following seller's strategy:
\begin{enumerate}
	\item The seller will provide $M-i+1$ arms with maximum value $v_i$ for $i = 1,...,m$. For notation convenience, label them as $i \cdot M + 1, i \cdot  M+2...,i\cdot  M + (M - i+1)$. 
	\item The $j$-th arm with maximum value $v_i$ (arm $i\cdot M + j$):
		\begin{enumerate}
			\item For rounds in $[(l-1)T/M+1, (l-1)T/M + 2\delta T]$, charge 0 and always give the item to the buyer for $l = 1,...,i-1$.
			\item For rounds in $[(l-1)T/M + 2\delta T+1, lT/M]$, charge $v_l$ and always give the item to the buyer for $l = 1,...,i-1$.
			\item For rounds in $[(i-1 + j)T/M+1, (i-1 + j)T/M + j\delta T]$, charge 0 and always give the item to the buyer.
			\item For rounds in $[(i-1 + j)T/M + j\delta T+1, (i+ j)T/M]$, charge $v_i$ and always give the item to the buyer.
			\item For other rounds, charge 0 and don't give the item to the buyer.
		\end{enumerate}
\end{enumerate}

Let $A_{i,j} = (i-1 + j)T/M + j\delta T+1$ and $B_{i,j} = (i+ j)T/M$. 
\begin{lemma}
For each $v_i \in \D$, $j = 1,...,M-i+1$, and round $\tau \in [A_{i,j},B_{i,j}]$, $\sigma_{i\cdot M + j,\tau}(v_i) > \sigma_{i' \cdot M + j', \tau}(v_i)$ for all $i' \leq i$ and $(i',j')\neq(i,j)$. 
\end{lemma}

\begin{proof}
First of all, by following the definition of the seller's strategy, for $\tau \in [A_{i,j},B_{i,j}]$, we have
\[
\sigma_{i\cdot M + j,\tau}(v_i) = v_i \cdot (2(i-1) + j)\delta T + \sum_{l = 1}^{i-1}  (v_i- v_l)(1/M-2\delta)T.
\]
There are several cases:
\begin{enumerate}
\item $i' \leq i$ and $j' +i ' > j + i$: for $\tau \in [A_{i,j},B_{i,j}]$, we have 
\begin{eqnarray*}
\sigma_{i'\cdot M + j' ,\tau}(v_i) &=& v_i \cdot (2(i'-1)) \delta T + \sum_{l = 1}^{i'-1}  (v_i- v_l)(1/M-2\delta)T \\
&\leq& v_i \cdot(2(i-1)) \delta T + \sum_{l = 1}^{i-1}  (v_i- v_l)(1/M-2\delta)T\\
&\leq& \sigma_{i\cdot M + j,\tau}(v_i)- v_i \cdot j \delta T \\
&<& \sigma_{i\cdot M + j,\tau}(v_i) - \gamma T.
\end{eqnarray*}
\item $i' = i$ and $j' < j$: for $\tau \in [A_{i,j},B_{i,j}]$, we have 
\begin{eqnarray*}
\sigma_{i'\cdot M + j' ,\tau}(v_i) &=& v_i \cdot (2(i-1)+j') \delta T + \sum_{l = 1}^{i-1}  (v_i- v_l)(1/M-2\delta)T \\
&\leq& \sigma_{i\cdot M + j,\tau}(v_i)- v_i \cdot (j-j')\delta T \\
&<& \sigma_{i\cdot M + j,\tau}(v_i) - \gamma T.
\end{eqnarray*}
\item $i' < i$ and $j' + i'\leq j +i $: for $\tau \in [A_{i,j},B_{i,j}]$, we have 
\begin{eqnarray*}
\sigma_{i'\cdot M + j' ,\tau}(v_i) &=& v_i \cdot (2(i'-1)+j') \delta T + \left(\sum_{l = 1}^{i'-1}  (v_i- v_l)(1/M-2\delta)T\right)  + (v_i - v_{i'})(1/M-j'\delta )T \\
&\leq& v_i \cdot (2(i'-1)+j') \delta T + \left( \sum_{l = 1}^{i-1}  (v_i- v_l)(1/M-2\delta)T \right) + v_i \cdot \max(2-j',0)\delta T \\
&\leq& v_i \cdot (2(i-1) + j -1) \delta T + \left( \sum_{l = 1}^{i-1}  (v_i- v_l)(1/M-2\delta)T \right) \\
&\leq& \sigma_{i\cdot M + j,\tau}(v_i)- v_i \delta T \\
&<& \sigma_{i\cdot M + j,\tau}(v_i) - \gamma T.
\end{eqnarray*}
\end{enumerate}
\end{proof}

It follows from the mean-based condition (Definition \ref{def:mb-cont}) that in the interval $ [A_{i,j},B_{i,j}]$ the buyer with value $v_i$ will, with probability at least $(1-Mm\gamma)$, choose arm $i \cdot M + j$. Since the buyer has value $v_i$ for the item with probability $q_i$, the total contribution of the buyer with value $v_i$ to the expected revenue of the seller is given by
\begin{eqnarray*}
q_i\cdot v_i\cdot \sum_{j=1}^{M-i+1}(1-mM\gamma)(B_{i,j}- A_{i,j}+1) &=&  q_i \cdot v_i\cdot (1-mM\gamma) (M-m)(T/M - m\delta T)	\\
&\geq& q_i \left( v_i\cdot T \cdot( 1- m/M) - o(T)\right) \\
&\geq& q_i \left( v_i\cdot T \cdot( 1- \varepsilon) - o(T)\right). \\
\end{eqnarray*}

Then we have that the expected revenue of the seller is at least

\begin{eqnarray*}
\sum_{i} q_i \left( v_i\cdot T \cdot( 1- \varepsilon) - o(T)\right) &=& (1-\eps)\left(\sum_{i}q_iv_i\right)T - o(T) \\
&=& (1-\eps)\E_{v\sim \D}[v] T - o(T)\\
&=& (1-\eps)\Val(\D)T - o(T).
\end{eqnarray*}

\end{proof}
 
\end{document}